\documentclass{article}

\usepackage{arxiv}
\usepackage[utf8]{inputenc}
\usepackage[T1]{fontenc}
\usepackage{hyperref}
\usepackage{url}
\usepackage{booktabs}
\usepackage{amsmath,amssymb,amsfonts}
\usepackage{amsthm}
\usepackage{mathrsfs}
\usepackage{nicefrac}
\usepackage{microtype}
\usepackage{graphicx}
\usepackage{multirow}
\usepackage{textcomp}
\usepackage{xcolor}
\usepackage{algorithm}
\usepackage{algorithmicx}
\usepackage{algpseudocode}
\usepackage{rotating}
\usepackage{pifont}
\usepackage{changepage}
\usepackage{listings}
\usepackage{tikz}
\usepackage{cite}
\usetikzlibrary{positioning,arrows.meta}

\graphicspath{{./images/}}

\newcommand{\cmark}{\ding{51}}
\newcommand{\xmark}{\ding{55}}
\newcommand{\sample}{\mathrel{\xleftarrow{\$}}}
\definecolor{headerblue}{RGB}{13,71,161}
\definecolor{rowgray}{RGB}{245,245,245}

\newtheorem{Theorem}{Theorem}
\newtheorem{Proposition}{Proposition}
\theoremstyle{definition}
\newtheorem{Definition}{Definition}
\theoremstyle{remark}
\newtheorem{Remark}{Remark}

\title{A Threshold Homomorphic Blockchain Architecture for Secure and Scalable IoT Sensor Data Aggregation}

\author{
Narendra Kumar Dewangan\\
LTCI, Department of Computer Science and Networks (INFRES),\\
Telecom Paris, Institut Polytechnique de Paris, France\\
\And
Mounira Msahli\\
LTCI, Department of Computer Science and Networks (INFRES),\\
Telecom Paris, Institut Polytechnique de Paris, France\\
\texttt{mounira.msahli@telecom-paris.fr}
}

\begin{document}
\maketitle

\begin{abstract}
Homomorphic-encryption blockchain frameworks for IoT sensor data aggregation typically rely on classical cryptographic hardness assumptions to provide on-chain data confidentiality, while rarely considering network topology as a factor affecting liveness and performance. This work introduces $\Phi$-PHE-BC, a topology-aware homomorphic blockchain architecture for secure, privacy-preserving IoT sensor data aggregation, employing threshold Paillier decryption and graph-parameterized performance analysis. Using a network graph $G$, device capability class $D$, aggregation function $\Phi$, cryptographic parameter set $\Lambda$, and consensus parameters $\Pi$, the framework defines a protocol instantiation whose security and performance properties are explicitly connected to the validator graph. We demonstrate that on-chain Paillier ciphertexts enable homomorphic aggregation of sensor readings while achieving standard classical security: IND-CPA security under the Decisional Composite Residuosity (DCR) assumption for aggregation confidentiality, and EUF-CMA security for transaction integrity via authentication signatures. Threshold partial-decryption shares are further protected by a noise-flooding wrapper whose privacy guarantee is information-theoretic under the configured statistical-hiding condition. Under partial synchrony and Byzantine fault-tolerance assumptions, protocol liveness is determined by the validator subgraph satisfying $\kappa(G_v) \geq f+1$. We further derive topology-parameterized bounds on throughput across tree, star, mesh, and scale-free deployment models, along with an explicit per-block communication-cost model. Our game-theoretic analysis shows that honest participation is a dominant strategy for each validator, yielding an all-honest Nash equilibrium under the stated utility model. The proposed architecture is practical, achieving consistently lower end-to-end latency than the selected traditional PHE-blockchain baseline while maintaining controllable threshold-decryption overhead in Hyperledger Fabric~2.5. The experimental configurations are consolidated in a benchmark table that distinguishes topology scaling, validator sensitivity, threshold decryption, and Byzantine-load experiments. The results indicate that $\Phi$-PHE-BC provides a practical solution for secure, privacy-preserving, and topology-aware IoT sensor data aggregation.
\end{abstract}


\section{Introduction}

Blockchain technology can provide transparency and privacy simultaneously;
however, data generated by Internet of Things (IoT) devices may contain highly
sensitive information about their owners and may be transmitted over open
communication channels. Modern distributed ledger platforms have introduced
lighter-weight consensus algorithms and more efficient implementations, yet
blockchain security remains challenged by the increasing computational
capabilities of adversaries and the continuous advancement of specialized
hardware for cryptographic attacks~\cite{10737075}. The Bitcoin
protocol~\cite{nakamoto2009bitcoin}, for example, employs the Proof-of-Work
(PoW) consensus mechanism, which requires substantial computational resources
and time for transaction confirmation and block generation.

The rapid proliferation of IoT deployments across critical infrastructure
domains, including industrial control systems, healthcare monitoring, smart
energy grids, and intelligent transportation networks, has generated
unprecedented volumes of sensitive operational data that must be processed
across distributed network environments~\cite{10737075,10755989,WANG2026107922}.
A fundamental tension characterizes this landscape: although the utility of IoT
data is maximized through aggregation and cross-organizational analysis, the
sensitivity of individual device readings requires raw data to remain
confidential, potentially even from infrastructure operators. Blockchain
technology has consequently emerged as a promising substrate for IoT data
management because of its tamper-resistant ledger, decentralized governance,
and programmable smart-contract capabilities~\cite{nakamoto2009bitcoin,XU2024124151}.
However, the transparency inherent in blockchain architectures---the very
property that enables auditability---directly conflicts with data confidentiality
requirements, thereby creating a privacy paradox that conventional blockchain
deployments cannot adequately resolve.

Partially Homomorphic Encryption (PHE), particularly the Paillier cryptosystem,
provides a principled approach to addressing this challenge by enabling
arithmetic computations directly over ciphertexts without requiring
decryption~\cite{SI202468,Paillier1999}. Under Paillier encryption, a blockchain
network can aggregate IoT sensor readings, compute weighted statistics, and
verify aggregated results while individual sensor values remain encrypted on
the distributed ledger. This approach has been explored in federated
learning~\cite{XIONG202495,10475694,10819476}, healthcare data
sharing~\cite{10737075,FIRDAUS2025101579}, carbon
accounting~\cite{HE2024110304}, cross-chain
interoperability~\cite{SI202468,Cao2024}, and transportation IoT
networks~\cite{10755989,10742617}. Fully Homomorphic Encryption (FHE)
schemes, including Brakerski/Fan-Vercauteren (BFV) and
Cheon--Kim--Kim--Song (CKKS), further extend the class of supported
computations to arbitrary polynomial operations and approximate floating-point
arithmetic~\cite{Paillier1999,WU2024467,LIU2025100689}, thereby enabling richer
analytics at the cost of substantially greater computational overhead.

Despite this body of work, two important limitations remain insufficiently
addressed in existing PHE-blockchain systems. First, long-lived IoT data
streams, including patient health records, energy-consumption patterns, and
industrial process parameters, may be stored and retained for extended periods.
This raises the broader question of how such deployments should evolve as
cryptographic standards advance. We revisit this issue in Section~\ref{subsec:limitations_future} as a
direction for future work rather than as a problem solved by the present
architecture. Second, and more importantly, the network topology of IoT
blockchain deployments, which fundamentally influences both security resilience
and system performance, is largely ignored in existing protocol designs and
analyses. Prior schemes commonly treat the validator and device network as an
abstract or ideal communication channel and consequently derive no guarantees
concerning liveness, throughput, or Byzantine fault tolerance as explicit
functions of the underlying graph structure.

In practical IoT deployments, network structures differ substantially across
application domains. Hierarchical tree topologies may arise in smart-grid
metering infrastructures, star configurations in hospital sensor networks, mesh
interconnections on industrial floors, and scale-free structures in smart-city
deployments. These topological differences directly affect communication paths,
connectivity, fault propagation, consensus delays, and resilience against
validator failures or adversarial disruption. Nevertheless, their security and
performance implications have not been systematically incorporated into existing
homomorphic blockchain architectures.

To address this gap, this paper presents $\Phi$-PHE-BC, a domain-agnostic and
topology-aware homomorphic aggregation framework for IoT blockchain networks.
The framework provides classical confidentiality for on-chain Paillier
aggregation, information-theoretic protection for threshold-decryption shares,
and standard signature-based transaction integrity, while deriving liveness,
throughput guarantees and communication-cost bounds as explicit functions of the deployment
topology. Given a network graph $G$, device capability class $D$, aggregation
function $\Phi$, cryptographic parameter set $\Lambda$, and consensus parameters
$\Pi$, the framework produces a protocol instantiation whose security guarantees
and performance bounds are parameterized by graph-theoretic properties of $G$.

The main contributions of this work are as follows:

\begin{itemize}
    \item We prove that protocol liveness under a Byzantine adversary
    $\mathcal{A}(\lambda,f,q_e,q_h)$ is characterized by the condition that
    the vertex connectivity of the validator subgraph satisfies
    $\kappa(G_v) \geq f+1$. We establish this result uniformly across the
    topology classes
    $\mathcal{T}=\{T_{\mathrm{tree}},T_{\mathrm{star}},T_{\mathrm{mesh}},
    T_{\mathrm{sf}}\}$. We further derive closed-form bounds for throughput
    $\mathrm{TPS}(G)$ and per-block communication cost parameterized by
    validator-network diameter, link latency, block size, and BFT voting cost
    (Theorem~10 and Proposition~11).

    \item We establish the classical and information-theoretic security
    foundations of the protocol stack. Specifically, we prove IND-CPA security
    of on-chain Paillier aggregation under the Decisional Composite Residuosity
    (DCR) assumption (Theorem~4), information-theoretic protection of threshold
    partial-decryption shares through noise flooding against computationally
    unbounded adversaries (Theorem~5), and EUF-CMA transaction integrity under
    standard digital-signature assumptions (Theorem~7).

    \item We prove that honest validator participation remains the dominant
    Nash-equilibrium strategy for validators and show that this conclusion is
    invariant to the magnitude of the per-transaction cryptographic cost
    multiplier. 

    \item We implement $\Phi$-PHE-BC on Hyperledger Fabric and evaluate
    throughput, end-to-end latency, and threshold-decryption overhead across
    four topology classes, benchmarking the proposed framework against the
    closest classical PHE-blockchain baselines.
\end{itemize}

\begin{table}[t]
\centering
\caption{Comparison with the closest prior schemes.}
\label{tab:comparison}
\renewcommand{\arraystretch}{1.15}
\begin{tabular}{lccc}
\hline
\textbf{Scheme} &
\textbf{Formal  Security Analysis} &
\textbf{Topology-Aware} &
\textbf{Threshold Decryption} \\
\hline
Lakhan et al.~\cite{10755989} & $\times$ & $\times$ & $\times$ \\
Jiang et al.~\cite{10742617} & $\times$ & $\times$ & $\times$ \\
Si et al.~\cite{SI202468} & $\times$ & $\times$ & $\checkmark$ \\
Wu et al.~\cite{WU2024467} & $\times$ & $\times$ & $\times$ \\
Latif et al.~\cite{Latif2026} & $\times$ & $\times$ & $\times$ \\
This work & $\checkmark$ & $\checkmark$ & $\checkmark$ \\
\hline
\end{tabular}

\vspace{0.5em}
\footnotesize
$\checkmark$ = property satisfied; $\times$ = property not satisfied.
\end{table}

Table~\ref{tab:comparison} positions $\Phi$-PHE-BC with respect to the closest
prior schemes. Lakhan et al.~\cite{10755989} integrate homomorphic encryption
with blockchain for transportation IoT applications but do not provide a
formal end-to-end security analysis or topology-aware analysis. Jiang
et al.~\cite{10742617} propose a DAG-blockchain PHE scheme for vehicular
networks without formally modeling the effects of network topology. Si
et al.~\cite{SI202468} employ threshold Paillier encryption for cross-chain
access control and represent the closest antecedent to our threshold-decryption
component; however, their framework assumes an ideal network and provides no
topology-parameterized guarantees. Wu et al.~\cite{WU2024467} integrate FHE into
an Ethereum-based environment but do not provide either a formal end-to-end security
analysis for the complete protocol stack or explicit network-topology analysis.
Latif et al.~\cite{Latif2026} present a blockchain-based authentication system
for e-health incorporating multi-factor authentication, hybrid RBAC--ABAC
access control, and post-quantum digital signatures; however, their work does
not address homomorphic aggregation or topology-aware guarantees for liveness,
throughput, and resilience.

Based on the literature reviewed in this study, we did not identify
prior work that jointly combines a formal end-to-end security treatment
for threshold homomorphic aggregation with topology-aware liveness and
throughput guarantees, together with connectivity-based fault-tolerance
and communication-cost analysis, within a single IoT blockchain
architecture. The remainder of this paper is organized as follows. Section~3 formalizes the network graph model,
device capability classes, adversary model, and security objectives.
Section~4 presents the $\Phi$-PHE-BC framework design, including the
cryptographic-layer selection algorithm and the Paillier--BFV bridge.
Section~5 establishes the end-to-end classical and information-theoretic security
guarantees of the protocol stack. Section~6 derives topology-parameterized
bounds for liveness, throughput, and communication cost. Section~7 analyzes
game-theoretic validator stability. Section~8 presents the implementation and
experimental evaluation. Finally, Section~9 concludes the paper.

\section{Literature Review}
\label{sec:literature}

A substantial body of work addresses the foundations, implementation, and
security of homomorphic encryption (HE). Paillier's additive cryptosystem
provides the classical PHE basis used in this paper~\cite{Paillier1999}, while
recent surveys and bibliometric studies map the evolution, capabilities, and
open challenges of modern HE systems~\cite{LIU2025100689,HARISS2026100815,SHARMA2026110969}.
Work on integrating FHE with blockchain and smart contracts highlights both the
expressiveness of schemes such as BFV/CKKS and their implementation
complexity~\cite{WU2024467,cryptoeprint:2025/527}. Public verification of FHE
computations, formal verification of FHE toolchains, homomorphic signatures,
and implementation-safety benchmarking further emphasize that practical HE
systems require more than encryption alone~\cite{cryptoeprint:2024/1764,Yang2025FormalVO,10854466,njungle2025safety}.
Selective and application-oriented HE designs continue to broaden the practical
space, including privacy-preserving face-recognition pipelines~\cite{11373303}.

A second major cluster combines HE with blockchain-assisted federated learning
and distributed analytics. CoPiFL studies collusion resistance and
privacy-preserving crowdsourcing~\cite{XIONG202495}; reputation and auditing
mechanisms are incorporated into blockchain-based FL in
\cite{10475694,11363238}; and PBFL employs homomorphic encryption with single
masking for privacy-preserving model aggregation~\cite{10819476}. Healthcare
and IoT deployments demonstrate the use of blockchain and HE in regulated,
sensitive-data settings~\cite{FIRDAUS2025101579,WANG2026107922,Jagdeesh2026},
while transportation-oriented FL and intrusion-detection designs show the same
pattern in highly distributed environments~\cite{11404228}. Collectively,
these works demonstrate the feasibility of blockchain--HE integration, but they
typically do not connect validator-network graph structure to formal liveness
and throughput guarantees.

HE--blockchain designs have also been investigated directly for IoT,
transportation, and cross-domain data sharing. Lakhan et al. propose a
homomorphic blockchain scheme for intelligent transport services in fog/cloud
and IoT networks~\cite{10755989}; Jiang et al. combine homomorphic encryption
with a DAG blockchain for vehicular networks~\cite{10742617}; and He et al.
apply blockchain-based protected accounting to carbon-emission data~\cite{HE2024110304}.
Cross-chain and searchable-encryption systems extend encrypted processing to
interoperable or attribute-controlled data-sharing settings~\cite{Cao2024,WANG2026112026}.
The blockchain-security survey of Xu et al. provides broader context for the
fusion of cryptographic protection with distributed ledgers~\cite{XU2024124151}.

Threshold decryption and privacy-preserving healthcare analysis are especially
relevant to the proposed architecture. Si et al. use threshold Paillier in a
cross-chain access-control mechanism~\cite{SI202468}; Guan et al. protect
healthcare-IoT data with blockchain-enhanced privacy mechanisms~\cite{10737075};
and PriCollabAnalysis combines blockchain, homomorphic encryption, secret
sharing, and secure multiparty computation for collaborative healthcare
analytics~\cite{Tawfik2025}. 

A smaller but important body of work considers HE beyond blockchain-centric IoT
aggregation. Ci et al. study privacy-preserving word-vector learning using PHE
\cite{CI2025103999}, while AlShaikh et al. combine partial homomorphic
processing with authenticated cloud-computing mechanisms~\cite{AlShaikh2025}.
These studies reinforce the broader utility of PHE for privacy-preserving
computation, but they do not address the topology-dependent BFT guarantees that
are central to the present work.

\subsection{Research Gaps and Motivation}
\label{sec:research_gaps}

The reviewed literature reveals four gaps that motivate $\Phi$-PHE-BC. First,
existing HE--blockchain systems generally select a single cryptographic
computation model for the full workflow. Additive PHE is attractive for frequent
linear aggregation because of its lower cost, whereas BFV/CKKS support richer
computations with greater overhead~\cite{LIU2025100689,WU2024467}. The proposed
architecture therefore uses Paillier as the primary on-chain aggregation path
and treats the TEE-assisted Paillier--BFV bridge as an optional extension for
non-linear functions rather than as a requirement for the evaluated baseline.

Second, prior IoT and vehicular blockchain schemes do not explicitly derive BFT
liveness and throughput from validator-graph properties. Transportation and
vehicular schemes such as~\cite{10755989,10742617,11404228} model distributed
operation but do not establish a necessary-and-sufficient validator-connectivity
condition of the form proved in Section~\ref{sec:network_resilience}. The present
work narrows the scope to this directly analyzed problem; dynamic membership and
mobility-aware rekeying remain outside the formal model.

Third, computation-validity mechanisms are often treated separately from the
complete aggregation workflow. Public-verification and SoK results identify
verification as an important open implementation dimension
\cite{cryptoeprint:2024/1764,cryptoeprint:2025/527}. Accordingly, the framework
includes an optional Fiat--Shamir ciphertext-validity proof for capable devices.
This mechanism proves the configured plaintext/randomness relation for a
submitted ciphertext; it is not claimed to provide a constant-time proof of
correctness for the complete aggregate.

Fourth, incentive alignment and network topology are rarely analyzed jointly in
HE-enabled blockchain systems. Blockchain-based FL schemes introduce reputation
or audit mechanisms~\cite{10475694,11363238}, but do not derive the validator
penalty condition studied in Section~\ref{sec:game}. Likewise, the reviewed
threshold-Paillier and IoT systems~\cite{SI202468,10737075,WANG2026107922} do not
parameterize liveness and throughput by validator vertex connectivity and
diameter. This omission is the central gap addressed in Section~\ref{sec:network_resilience}.

For future cryptographic migration, the standardized ML-KEM and ML-DSA
primitives provide relevant post-quantum key-establishment and signature
candidates~\cite{NIST-FIPS203,NIST-FIPS204}; SLH-DSA provides an additional
standardized hash-based signature option~\cite{nistFederalInformation}. These
standards are cited only as future-work candidates and are not implemented or
claimed as security properties of the present $\Phi$-PHE-BC system.

\section{System Model}
\label{sec:model}

The list of symbols used in this paper is given in
Table~\ref{tab:glossary} in the Appendix.

\subsection{Network Graph Model}
\label{subsec:graph}

We model the IoT blockchain network as a weighted undirected graph.

\begin{Definition}[IoT Blockchain Network Graph]
\label{def:network}
An IoT blockchain network is a tuple
\[
\mathcal{G}=(\mathcal{V},\mathcal{E},W),
\]
where:
\begin{itemize}
    \item $\mathcal{V}=\mathcal{V}_d\cup\mathcal{V}_v\cup\mathcal{V}_p$ is a
    finite node set partitioned into IoT device nodes $\mathcal{V}_d$,
    validator nodes $\mathcal{V}_v$, and blockchain peer nodes
    $\mathcal{V}_p$, with
    \[
    |\mathcal{V}_d|=n_d,\qquad
    |\mathcal{V}_v|=n_v,\qquad
    |\mathcal{V}_p|=n_p;
    \]

    \item $\mathcal{E}\subseteq\mathcal{V}\times\mathcal{V}$ is a set of
    authenticated point-to-point channels; and

    \item $W:\mathcal{E}\rightarrow\mathbb{R}_{>0}$ assigns a propagation
    latency $\tau_{uv}$ to each link $(u,v)\in\mathcal{E}$.
\end{itemize}
\end{Definition}

The validator subgraph
\[
\mathcal{G}_v=\mathcal{G}[\mathcal{V}_v]
\]
is the subgraph induced by the validator nodes. The maximum number of Byzantine
validators tolerated by the protocol is
\[
f=\left\lfloor\frac{n_v-1}{3}\right\rfloor,
\]
consistent with standard BFT assumptions~\cite{nakamoto2009bitcoin}.

Two graph-theoretic quantities govern the security and performance analysis
throughout this paper. The \emph{vertex connectivity}
$\kappa(\mathcal{G}_v)$ is the minimum number of validator nodes whose removal
disconnects $\mathcal{G}_v$; it characterizes fault tolerance. The
\emph{algebraic connectivity} (Fiedler value) $\lambda_2(\mathcal{G}_v)$ is the
second-smallest eigenvalue of the Laplacian of $\mathcal{G}_v$; it quantifies
how quickly information spreads across the validator subgraph.

We consider four topology classes that cover the principal IoT deployment
archetypes identified in the literature~\cite{10755989,10737075}:
\begin{itemize}
    \item $\mathcal{T}_{\mathrm{tree}}$: hierarchical tree, with
    \[
    \operatorname{diam}(\mathcal{G})=O(\log n_d),
    \qquad
    \kappa(\mathcal{G}_v)\geq 1,
    \]
    representing smart-grid metering hierarchies;

    \item $\mathcal{T}_{\mathrm{star}}$: hub-and-spoke topology, with
    \[
    \operatorname{diam}(\mathcal{G})=2,
    \qquad
    \kappa(\mathcal{G}_v)=1,
    \]
    representing hospital sensor networks;

    \item $\mathcal{T}_{\mathrm{mesh}}$: regular grid or torus, with
    \[
    \operatorname{diam}(\mathcal{G})=O(\sqrt{n_v}),
    \qquad
    \kappa(\mathcal{G}_v)\geq 4,
    \]
    representing industrial-floor sensor meshes; and

    \item $\mathcal{T}_{\mathrm{sf}}$: Barab\'asi--Albert scale-free graph,
    with
    \[
    \operatorname{diam}(\mathcal{G})=O(\log n_v),
    \]
    and variable $\kappa(\mathcal{G}_v)$, representing smart-city
    deployments.
\end{itemize}

\begin{Remark}
Star and tree topologies have $\kappa(\mathcal{G}_v)=1$, which violates the
liveness condition derived in Section~6. In such deployments, the validator
subgraph must be augmented independently of the device-layer topology so that
the required condition
\[
\kappa(\mathcal{G}_v)\geq f+1
\]
is satisfied. The device and validator sub-topologies are therefore treated
separately throughout the paper.

We assume authenticated-path routing among validators: a message between two
validators may be forwarded along a path in $\mathcal{G}_v$, and an honest
validator forwards correctly received protocol messages. Consequently, after
the removal of Byzantine validators, connectivity of the remaining honest
validator subgraph implies mutual communication among honest validators.
\end{Remark}

\subsection{Device Capability Model}
\label{subsec:device}

IoT devices span a wide range of computational resources. We classify them by
sustained cryptographic throughput, following classifications used in related
IoT security work~\cite{FIRDAUS2025101579,10737075}.

\begin{Definition}[Device Capability Class]
\label{def:device}
A device $d\in\mathcal{V}_d$ belongs to a class
\[
\mathcal{D}(d)\in
\left\{
\mathcal{D}_{\mathrm{I}},
\mathcal{D}_{\mathrm{II}},
\mathcal{D}_{\mathrm{III}}
\right\}
\]
according to:
\begin{itemize}
    \item \textbf{$\mathcal{D}_{\mathrm{I}}$ (constrained):}
    CPU $<100$~MHz and RAM $<256$~KB. Supports Paillier encryption only and
    cannot execute the full signature and zero-knowledge-proof stack on-device.

    \item \textbf{$\mathcal{D}_{\mathrm{II}}$ (moderate):}
    $100$~MHz $\leq$ CPU $<1$~GHz and RAM $\geq1$~MB. Supports Paillier
    encryption, session-key establishment, and standard digital signatures for
    transaction authentication.

    \item \textbf{$\mathcal{D}_{\mathrm{III}}$ (capable):}
    CPU $\geq1$~GHz and RAM $\geq64$~MB. Supports the full cryptographic layer
    stack, including standard digital-signature generation and zero-knowledge-proof generation.
\end{itemize}
\end{Definition}

The device class directly drives the selection of the cryptographic-layer
algorithm presented in Section~4.3. We note that the specific signature and
key-establishment primitives assigned to
$\mathcal{D}_{\mathrm{II}}$ and $\mathcal{D}_{\mathrm{III}}$ devices are
implementation choices rather than fixed requirements of the framework.

\subsection{Adversary Model}
\label{subsec:adversary}

We define a parameterized Byzantine adversary consistent with standard
partial-synchrony BFT analysis.

\begin{Definition}[Byzantine Adversary]
\label{def:adversary}
The adversary
\[
\mathcal{A}=\mathcal{A}(\lambda,f,q_e,q_h)
\]
is characterized by:
\begin{itemize}
    \item a security parameter $\lambda\in\mathbb{N}$;

    \item Byzantine control of at most $f$ validator nodes, which may collude
    and deviate arbitrarily from the protocol;

    \item at most $q_e$ encryption-oracle queries and $q_h$ random-oracle
    (hash) queries; and

    \item full Dolev--Yao control of all network channels:
    $\mathcal{A}$ may eavesdrop on, replay, delay, and inject messages on any
    link $(u,v)\in\mathcal{E}$.
\end{itemize}

The adversary $\mathcal{A}$ is computationally bounded in the standard sense
used for the classical hardness assumptions invoked in Section~5, namely the
Decisional Composite Residuosity (DCR) assumption for Paillier encryption and
the hardness assumption underlying the chosen digital signature scheme.
\end{Definition}

\begin{table*}[t]
\caption{Notation and symbols used in the system model.}
\label{tab:glossary}
\centering
\renewcommand{\arraystretch}{1.15}
\begin{tabular}{ll}
\hline
\textbf{Symbol} & \textbf{Description} \\
\hline
$\mathcal{G}=(\mathcal{V},\mathcal{E},W)$ &
Weighted IoT blockchain network graph \\
$\mathcal{V}$ & Set of all network nodes \\
$\mathcal{V}_d$ & Set of IoT device nodes \\
$\mathcal{V}_v$ & Set of validator nodes \\
$\mathcal{V}_p$ & Set of blockchain peer nodes \\
$n_d$ & Number of IoT device nodes, $|\mathcal{V}_d|$ \\
$n_v$ & Number of validator nodes, $|\mathcal{V}_v|$ \\
$n_p$ & Number of blockchain peer nodes, $|\mathcal{V}_p|$ \\
$\mathcal{E}$ & Set of authenticated communication links \\
$W$ & Edge-weight function representing propagation latency \\
$\tau_{uv}$ & Propagation latency of link $(u,v)$ \\
$\mathcal{G}_v$ & Validator-induced subgraph $\mathcal{G}[\mathcal{V}_v]$ \\
$f$ & Maximum number of tolerated Byzantine validators \\
$\kappa(\mathcal{G}_v)$ & Vertex connectivity of the validator subgraph \\
$\lambda_2(\mathcal{G}_v)$ & Algebraic connectivity (Fiedler value) of $\mathcal{G}_v$ \\
$\operatorname{diam}(\mathcal{G})$ & Diameter of graph $\mathcal{G}$ \\
$\mathcal{T}_{\mathrm{tree}}$ & Hierarchical tree topology class \\
$\mathcal{T}_{\mathrm{star}}$ & Hub-and-spoke topology class \\
$\mathcal{T}_{\mathrm{mesh}}$ & Mesh/grid topology class \\
$\mathcal{T}_{\mathrm{sf}}$ & Scale-free topology class \\
$\mathcal{D}(d)$ & Capability class of device $d$ \\
$\mathcal{D}_{\mathrm{I}}$ & Constrained device capability class \\
$\mathcal{D}_{\mathrm{II}}$ & Moderate device capability class \\
$\mathcal{D}_{\mathrm{III}}$ & Capable device capability class \\
$\mathcal{A}$ & Byzantine adversary \\
$\lambda$ & Cryptographic security parameter \\
$\varepsilon_{\mathrm{nf}}$ & Statistical-hiding target for the flooding layer \\
$q_e$ & Number of encryption-oracle queries \\
$q_h$ & Number of random-oracle/hash queries \\
\hline
\end{tabular}
\end{table*}
\section{Framework Design}
\label{sec:framework}

\subsection{Framework Parameterization}
\label{subsec:parameterization}

\begin{Definition}[$\Phi$-PHE-BC Framework Instance]
\label{def:framework}
A $\Phi$-PHE-BC instance is a tuple
\[
\mathcal{F}
=
(\mathcal{G},\mathcal{D},\Phi,\Lambda,\Pi),
\]
where:
\begin{itemize}
    \item $\mathcal{G}$ is a network graph as defined in
    Definition~\ref{def:network};

    \item $\mathcal{D}\in
    \{\mathcal{D}_{\mathrm{I}},
    \mathcal{D}_{\mathrm{II}},
    \mathcal{D}_{\mathrm{III}}\}$
    is the device capability class as defined in
    Definition~\ref{def:device};

    \item
    \[
    \Phi:\mathbb{Z}^{n_d}\rightarrow\mathbb{Z}
    \]
    is the target aggregation function, classified as linear,
    polynomial of degree $\ell$, or floating-point;

    \item
    \[
    \Lambda=
    (n,N_{\mathrm{BFV}},q,t,\sigma,\lambda,\varepsilon_{\mathrm{nf}})
    \]
    is the cryptographic parameter tuple, where $n$ is the Paillier
    modulus, $N_{\mathrm{BFV}}$ is the BFV ring dimension, $q$ is the
    ciphertext modulus, $t$ is the plaintext modulus, $\sigma$ is the
    Gaussian standard deviation, and $\lambda$ is the security parameter, while $\varepsilon_{\mathrm{nf}}$ is the configured statistical-hiding target for the flooding layer;

    \item
    \[
    \Pi=(f,\tau_{\mathrm{thresh}},B,\tau_{\mathrm{ttl}})
    \]
    specifies the Byzantine threshold $f$, the number of required threshold
    shares $\tau_{\mathrm{thresh}}$, the block size $B$, and the nonce
    time-to-live $\tau_{\mathrm{ttl}}$.
\end{itemize}
\end{Definition}

The key design principle of $\Phi$-PHE-BC is that every security guarantee and
performance bound derived in Sections~5--6 is an explicit function of the
components of $\mathcal{F}$ and, in particular, of the graph-theoretic
properties of $\mathcal{G}$.

\subsection{Cryptographic Primitives}
\label{subsec:primitives}

The framework comprises the following cryptographic building blocks. Full
parameter specifications are given by $\Lambda$.

\paragraph{Paillier PHE.}
Let $n=pq$, where $p$ and $q$ are primes, and let $g=n+1$. Define
\[
\lambda_{\mathrm{P}}
=
\operatorname{lcm}(p-1,q-1),
\]
\[
\mu
=
L\!\left(g^{\lambda_{\mathrm{P}}}\bmod n^2\right)^{-1}
\bmod n,
\]
where
\[
L(x)=\frac{x-1}{n}.
\]

The Paillier encryption of a message $m\in\mathbb{Z}_n$ is
\begin{equation}
\operatorname{Enc}(m;r)
=
g^m r^n
\bmod n^2,
\qquad
r\sample \mathbb{Z}_n^{*},
\label{eq:paillier-enc}
\end{equation}
and decryption is
\begin{equation}
\operatorname{Dec}(c)
=
L\!\left(c^{\lambda_{\mathrm{P}}}\bmod n^2\right)
\mu
\bmod n.
\label{eq:paillier-dec}
\end{equation}

Homomorphic addition follows from
\[
\operatorname{Dec}(c_1c_2\bmod n^2)
=
m_1+m_2,
\]
and scalar multiplication follows from
\[
\operatorname{Dec}(c^k\bmod n^2)
=
km.
\]
The confidentiality of Paillier encryption is based on the Decisional
Composite Residuosity (DCR) assumption.

\paragraph{Transaction Authentication.}
Each device signs its transaction using a standard EUF-CMA-secure digital
signature scheme. The signature binds the device identity, ciphertext,
weight, timestamp, and nonce, as specified in Section~\ref{subsec:protocol}.
The specific signature scheme is an implementation parameter of $\Lambda$ and
is treated abstractly in the security analysis. Section~\ref{subsec:limitations_future} discusses a
possible migration to lattice-based signatures as future work; such a
migration is not an achieved security property of the present architecture.

\paragraph{BFV Fully Homomorphic Encryption.}
BFV operates in the polynomial ring
\[
R_q
=
\mathbb{Z}_q[X]/
(X^{N_{\mathrm{BFV}}}+1).
\]
Encryption embeds a plaintext $m\in\mathbb{Z}_t$ using
\[
\Delta
=
\left\lfloor
\frac{q}{t}
\right\rfloor.
\]
A ciphertext may be represented as
\[
c=
\left(
bu+e_1+\Delta m,\;
au+e_2
\right),
\]
where the error and masking terms are sampled according to the configured BFV
distributions.

BFV supports both homomorphic addition and multiplication, enabling evaluation
of aggregation functions of polynomial degree up to
\[
L_{\max}
=
\left\lfloor
\frac{
\log_2(q/t)
}{
\log_2(B_{\chi}+1)
}
\right\rfloor,
\]
where
\[
B_{\chi}=6\sigma
\]
is the configured noise bound.

\paragraph{Noise Flooding.}
The threshold-share privacy layer is modeled through a configurable flooding
transformation. For validator $V_i$, let $e_i$ denote the secret-dependent
exponent used by the threshold-decryption procedure and let $F_{\mathrm{nf}}$
be the configured integer-valued flooding distribution. For any two admissible
exponents $e,e'$ in the configured exponent domain, the deployment parameters
must satisfy the statistical-hiding condition
\begin{equation}
\operatorname{SD}\!\left(F_{\mathrm{nf}}+e,\,F_{\mathrm{nf}}+e'\right)
\leq \varepsilon_{\mathrm{nf}},
\label{eq:nf_requirement}
\end{equation}
where $\varepsilon_{\mathrm{nf}}$ is an explicit deployment parameter. This
condition is information-theoretic: it is stated directly in statistical
distance and does not invoke a computational hardness assumption. The present
paper does not claim a universal closed-form value of
$\varepsilon_{\mathrm{nf}}$ from a Gaussian width alone; parameter selection
must be justified for the concrete flooding construction used by an
implementation.

\subsection{Cryptographic Layer Selection}
\label{subsec:layerselect}

A key contribution of $\Phi$-PHE-BC is the \textsc{LayerSelect} algorithm,
which automatically determines the minimal cryptographic layer set required to
achieve the security goals defined in Section~5 for a given device class and
aggregation function. This avoids both under-provisioning, which may leave a
required protection mechanism absent, and over-provisioning, which may impose
unnecessary cryptographic overhead on constrained devices.

\begin{algorithm}[t]
\caption{\textsc{LayerSelect}$(\mathcal{D},\Phi,\lambda)$}
\label{alg:layerselect}
\begin{algorithmic}[1]
\Require Device class $\mathcal{D}$, aggregation function $\Phi$, security parameter $\lambda$
\Ensure Ordered cryptographic layer set $\mathcal{L}$
\State $\mathcal{L}\gets\{\textsc{Paillier},\textsc{NoiseFlooding}\}$
\If{$\Phi\notin\mathcal{F}_{\mathrm{linear}}$}
    \State $\mathcal{L}\gets\mathcal{L}\cup\{\textsc{BFV-Bridge}\}$
\EndIf
\If{$\Phi$ requires floating-point arithmetic}
    \State $\mathcal{L}\gets\mathcal{L}\cup\{\textsc{CKKS}\}$
\EndIf
\If{$\mathcal{D}\geq\mathcal{D}_{\mathrm{II}}$}
    \State $\mathcal{L}\gets\mathcal{L}\cup\{\textsc{Signature-Auth}\}$
\EndIf
\If{$\mathcal{D}=\mathcal{D}_{\mathrm{III}}$}
    \State $\mathcal{L}\gets\mathcal{L}\cup\{\textsc{ZKP}\}$
\EndIf
\State \Return $\mathcal{L}$
\end{algorithmic}
\end{algorithm}

\begin{Proposition}[Completeness of \textsc{LayerSelect}]
\label{prop:layerselect}
For any $(\mathcal{D},\Phi,\lambda)$, the layer set
\[
\mathcal{L}
=
\textsc{LayerSelect}(\mathcal{D},\Phi,\lambda)
\]
is sufficient to instantiate the protection mechanisms required by the
framework security goals.
\end{Proposition}

\begin{proof}[Proof sketch]
Transaction integrity is provided through an EUF-CMA-secure signature mechanism
whenever $\mathcal{D}\geq\mathcal{D}_{\mathrm{II}}$, while the constrained
device class uses the framework's corresponding ciphertext-binding and
validation mechanisms. Confidentiality of individual readings is provided by
Paillier encryption under the DCR assumption, while noise flooding protects the
published threshold-decryption shares under the statistical assumptions of the
protocol.

The BFV bridge is required only when
\[
\Phi\notin\mathcal{F}_{\mathrm{linear}},
\]
because Paillier alone supports additive operations and scalar multiplication
but does not directly support ciphertext--ciphertext multiplication required
for general polynomial functions of degree at least two.
\end{proof}

\subsection{Protocol Phases}
\label{subsec:protocol}

A $\Phi$-PHE-BC instance executes the following nine phases. The complete
Hyperledger Fabric chaincode realization is described in Section~8.

\paragraph{Phase 1: Device Registration.}
Each device $d_i\in\mathcal{V}_d$ generates a signature key pair
\[
(sk_i,pk_i)\leftarrow\operatorname{KeyGen}()
\]
and submits a certificate-signing request to the registration authority (RA).
The RA issues
\[
\operatorname{Cert}_i
=
\operatorname{Sign}_{\mathrm{RA}}(pk_i,ID_i),
\]
and the chaincode stores $\operatorname{Cert}_i$ under the key
\texttt{device:}$ID_i$ in the Fabric world state.

\paragraph{Phase 2: Sensing and Encryption.}
Device $d_i$ reads a sensor value
\[
v_i\in[0,n-1]
\]
and computes
\[
ct_i
=
\operatorname{Enc}(v_i;r_i)
\]
according to Equation~\eqref{eq:paillier-enc}. For
$\mathcal{D}_{\mathrm{II}}$ and $\mathcal{D}_{\mathrm{III}}$ devices, a
session key $K_{ij}$ may be established with the target validator $V_j$ using
an authenticated key-establishment mechanism. The ciphertext is then
transmitted inside the corresponding protected communication channel.

\paragraph{Phase 3: Transaction Submission.}
Device $d_i$ constructs
\[
T_i
=
(ID_i,ct_i,w_i,ts,\mathsf{nonce},\sigma_i),
\]
where $w_i$ is an integer-encoded weight,
\[
\mathsf{nonce}\sample\{0,1\}^{128}
\]
is a fresh nonce, and
\[
\sigma_i
=
\operatorname{Sign}_{sk_i}
\left(
H(ID_i\Vert ct_i\Vert\mathsf{nonce})
\right).
\]

\paragraph{Phase 4: Validator Verification.}
Upon receiving $T_i$, validator $V_j$ verifies:
\begin{enumerate}
    \item
    \[
    \operatorname{Verify}
    \left(
    pk_i,\sigma_i,
    H(ID_i\Vert ct_i\Vert\mathsf{nonce})
    \right)=1;
    \]

    \item $\operatorname{Cert}_i$ is valid under the RA public key;

    \item
    \[
    \mathsf{nonce}\notin\mathcal{N},
    \]
    where $\mathcal{N}$ denotes the set of previously recorded nonces; and

    \item
    \[
    |ts-t_{\mathrm{now}}|
    \leq
    \tau_{\mathrm{ttl}}.
    \]
\end{enumerate}

Transactions satisfying all checks are forwarded to the block-proposal phase.
The verification procedure is given in Algorithm~\ref{alg:verifytransaction}.

\paragraph{Phase 5: On-Chain Homomorphic Aggregation.}
The leader validator invokes the aggregation chaincode to compute
\[
ct_{\mathrm{agg}}
=
\prod_{i=1}^{n_d}
ct_i^{w_i}
\bmod n^2
\equiv
\operatorname{Enc}
\left(
\sum_{i=1}^{n_d}w_iv_i
\right).
\]
The computation is performed over ciphertexts, and individual plaintext sensor
readings are not stored in the blockchain world state.

\paragraph{Phase 6: BFT Consensus.}
The leader broadcasts the proposed block
\[
B_k=(H_k,T_1,\ldots,T_{n_d}),
\]
where
\[
H_k=
\left(
H(B_{k-1}),
\operatorname{MerkleRoot}(T_1,\ldots,T_{n_d}),
ts,
ct_{\mathrm{agg}},
\sigma_{\mathrm{leader}}
\right).
\]

Validators execute the configured PBFT-style consensus procedure. Under the
stated protocol assumptions, a block is committed after the required prepare
and commit quorums are obtained.

\paragraph{Phase 7: Threshold Decryption.}
The Paillier private key is distributed among validators using a
degree-$(\tau_{\mathrm{thresh}}-1)$ polynomial over the configured field. Each
validator $V_i$ first computes the ordinary threshold-decryption exponent
\[
e_i=2\delta sk_i,\qquad \delta=n_v!,
\]
and the deployment's flooding wrapper produces the published share transcript
from an exponent distributed as $e_i+F_{\mathrm{nf}}$. Reconstruction combines
at least $\tau_{\mathrm{thresh}}$ valid shares using the corresponding
integer-scaled Lagrange coefficients. Theorem~\ref{thm:noise_flooding} concerns
the privacy of the published flooded transcript only. The flooding wrapper is implemented using a correctness-preserving
mask-cancellation mechanism. Let the flooded exponent used by validator
$V_i$ be
\begin{equation}
\widetilde{e}_i = e_i + f_i,
\end{equation}
where $e_i = 2\delta s k_i$ is the ordinary threshold-decryption
exponent and $f_i$ is the locally generated flooding mask.

The published partial-decryption share is
\begin{equation}
\widetilde{d}_i
=
ct_{\mathrm{agg}}^{\,\widetilde{e}_i}
\bmod n^2.
\end{equation}

Before final plaintext recovery, the reconstruction procedure removes
the aggregate masking contribution according to the implemented
mask-cancellation rule. Denoting the integer-scaled Lagrange
coefficient of validator $V_i$ by $\lambda_i$, the combined masking
term is
\begin{equation}
F_{\mathrm{comb}}
=
\sum_{i\in S}
\lambda_i f_i,
\end{equation}
where $S$ is the set of participating validators.

The implementation cancels or removes the contribution corresponding
to $F_{\mathrm{comb}}$ before applying the final Paillier recovery
function. Consequently, the reconstructed value is identical to that
obtained from the unflooded threshold shares, while the individually
published shares remain statistically hidden according to the
condition of Theorem~\ref{thm:noise_flooding}.

\paragraph{Phase 8: ZKP Verification.}
For deployments involving $\mathcal{D}_{\mathrm{III}}$ devices, a device may
generate a Fiat--Shamir proof $\pi_i$ demonstrating knowledge of
$(v_i,r_i)$ such that
\[
\operatorname{Enc}(v_i;r_i)=ct_i.
\]
The aggregation chaincode verifies the configured proofs before aggregation,
thereby rejecting malformed ciphertexts according to the protocol rules.

\paragraph{Phase 9: Result Retrieval.}
An authorized querier invokes
\[
\operatorname{QueryLedger}(ID,h_{\mathrm{lo}},h_{\mathrm{hi}}),
\]
which returns the committed result and any associated verification proof,
enabling offline verification without repeating the complete decryption
procedure.

\begin{algorithm}[t]
\caption{Transaction Verification}
\label{alg:verifytransaction}
\begin{algorithmic}[1]
\Require Signed transaction $T_i=(ID_i,ct_i,w_i,ts,\mathsf{nonce},\sigma_i)$; nonce set $\mathcal{N}$; RA public key $pk_{\mathrm{RA}}$; current time $t_{\mathrm{now}}$; nonce time-to-live $\tau_{\mathrm{ttl}}$
\Ensure $T_i$ admitted to the proposal pool or rejected with an error
\If{$\operatorname{Verify}(pk_i,\sigma_i,H(ID_i\Vert ct_i\Vert\mathsf{nonce}))\neq1$}
    \State \Return \texttt{ERR\_INVALID\_SIGNATURE}
\EndIf
\If{$\operatorname{VerifyRA}(pk_{\mathrm{RA}},\operatorname{Cert}_i)\neq1$}
    \State \Return \texttt{ERR\_INVALID\_CERTIFICATE}
\EndIf
\If{$\mathsf{nonce}\in\mathcal{N}$}
    \State \Return \texttt{ERR\_REPLAY}
\EndIf
\If{$|ts-t_{\mathrm{now}}|>\tau_{\mathrm{ttl}}$}
    \State \Return \texttt{ERR\_STALE\_TIMESTAMP}
\EndIf
\State $\mathcal{N}\gets\mathcal{N}\cup\{\mathsf{nonce}\}$
\State \Return \texttt{ADMITTED}
\end{algorithmic}
\end{algorithm}

\subsection{The Paillier--BFV Bridge}
\label{subsec:bridge}

When the target aggregation function is not expressible using Paillier's
additive homomorphism, namely when
\[
\Phi\notin\mathcal{F}_{\mathrm{linear}},
\]
the \textsc{LayerSelect} algorithm adds the BFV bridge to $\mathcal{L}$.

\begin{Definition}[Paillier--BFV Bridge]
\label{def:bridge}
The bridge
\[
\mathcal{B}
=
(\operatorname{Gen},\operatorname{ReEnc},\operatorname{Verify})
\]
consists of:
\begin{itemize}
    \item \textbf{Gen:} executed inside a trusted execution environment (TEE)
    and producing a re-encryption key
    \[
    rk_{\mathrm{P}\rightarrow\mathrm{B}}
    =
    \operatorname{BFV.Enc}
    (H(sk_{\mathrm{P}});pk_{\mathrm{B}});
    \]

    \item \textbf{ReEnc:} given a Paillier ciphertext $ct_{\mathrm{P}}$,
    decrypting it inside the TEE boundary and producing
    \[
    ct_{\mathrm{B}}
    =
    \operatorname{BFV.Enc}
    \left(
    \operatorname{Paillier.Dec}(ct_{\mathrm{P}});
    pk_{\mathrm{B}}
    \right);
    \]

    \item \textbf{Verify:} a zero-knowledge verification mechanism attesting
    that the configured re-encryption procedure was executed correctly without
    revealing the underlying plaintext.
\end{itemize}
\end{Definition}

\begin{Proposition}[Bridge Necessity]
\label{prop:bridge-necessity}
There exists an aggregation function $\Phi$ of polynomial degree at least two
that is not computable under Paillier alone. Consequently, a mechanism enabling
transition to a multiplication-capable homomorphic scheme is required for
complete support of such functions.
\end{Proposition}

\begin{proof}
Paillier supports additive homomorphism. For ciphertexts $c_1$ and $c_2$,
\[
c_1c_2\bmod n^2
\]
decrypts to $m_1+m_2$, but this operation does not directly produce an
encryption of $m_1m_2$. Thus,
\[
\Phi(m_1,m_2)=m_1m_2
\]
is not directly computable using Paillier's additive operations alone. BFV
supports homomorphic multiplication and can therefore evaluate such
polynomial computations subject to its parameter and noise budget. The bridge
enables the transition of the protected value to the BFV computation domain.
\end{proof}

\begin{Proposition}[Bridge IND-CPA Security Conditional on TEE Integrity]
\label{prop:bridge-security}
If the TEE hardware and attestation mechanism are not compromised by the
out-of-scope attack classes defined in the TEE threat model, and the BFV
encryption scheme satisfies IND-CPA security under its underlying lattice-based
hardness assumption, then the bridge provides IND-CPA confidentiality against
the adversary model of Definition~\ref{def:adversary}, subject to the stated
TEE trust assumption.
\end{Proposition}

\begin{proof}[Proof sketch]
The output $ct_{\mathrm{B}}$ is a BFV encryption under $pk_{\mathrm{B}}$.
The intermediate plaintext
\[
m=
\operatorname{Paillier.Dec}(ct_{\mathrm{P}})
\]
remains within the trusted execution boundary. The adversary observes only the
resulting ciphertext and the configured verification proof. Under the stated
zero-knowledge and encryption-security assumptions, these observations do not
reveal the plaintext beyond the information permitted by the protocol.
\end{proof}

\paragraph{TEE Threat Model.}
The Paillier--BFV bridge relies on a trusted execution environment, such as an
Intel SGX- or ARM TrustZone-class environment, to perform
\textsc{ReEnc} without exposing the intermediate plaintext. The in-scope
adversary includes a network-level adversary with Dolev--Yao control of
communication channels and a compromised host operating system, hypervisor, or
Fabric peer process attempting to inspect protected enclave memory, intercept
the re-encryption process, or tamper with the bridge control flow.

The following attacks are outside the security guarantee of the bridge:
classical enclave side-channel attacks, including cache-timing, power-analysis,
and speculative-execution attacks; physical fault injection or voltage
glitching targeting the enclave hardware; and rollback attacks that replay
stale enclave state when no trusted monotonic counter or equivalent rollback
protection is available. Proposition~\ref{prop:bridge-security} is therefore
conditional on the integrity and confidentiality properties provided by the
TEE and its attestation mechanism against the assumed attack model.

\subsection{System Architecture Overview}
\label{subsec:architecture}

Figure~\ref{fig:architecture} illustrates the end-to-end system architecture
of a $\Phi$-PHE-BC instance. IoT devices form the bottom layer and encrypt
sensor readings using Paillier before submitting authenticated transactions.
Validator nodes form the middle layer: they verify transactions, execute
on-chain homomorphic aggregation through Fabric chaincode, participate in BFT
consensus, and contribute threshold-decryption shares. The Fabric ordering
service and peer nodes form the blockchain layer and persist committed blocks
and the encrypted world state. The optional Paillier--BFV bridge is activated
only when
\[
\Phi\notin\mathcal{F}_{\mathrm{linear}},
\]
and operates within a TEE associated with a designated validator. Result
consumers interact through the retrieval interface and obtain the committed
aggregate together with any configured verification evidence.

\begin{figure*}[t]
    \centering
    \includegraphics[width=\textwidth]{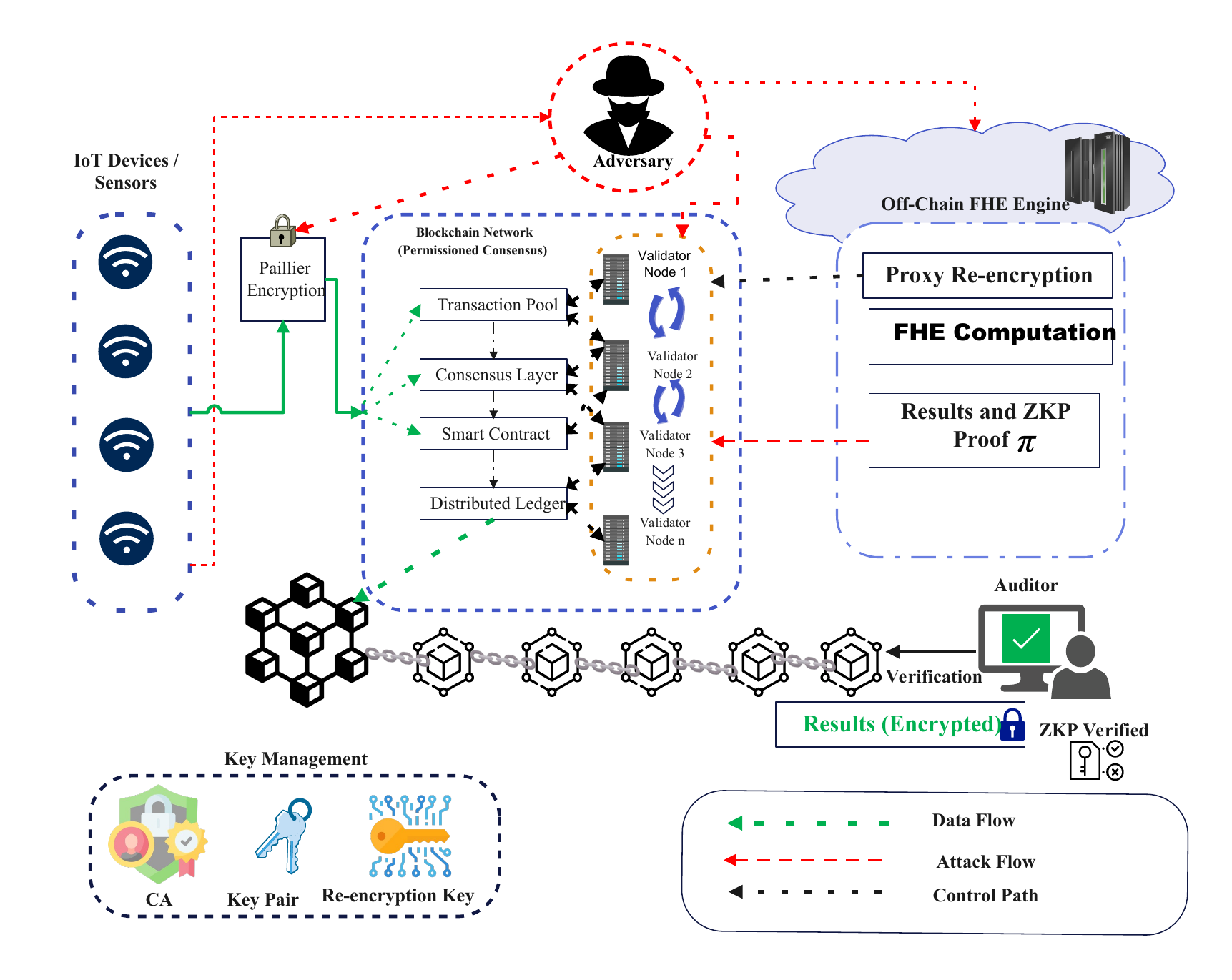}
    \caption{End-to-end system architecture of the proposed
    $\Phi$-PHE-BC framework, showing the IoT device layer, validator layer,
    Hyperledger Fabric blockchain layer, optional TEE-based Paillier--BFV
    bridge, threshold-decryption process, and result-retrieval interface.}
    \label{fig:architecture}
\end{figure*}

Figure~\ref{fig:workflow} shows the operation of the proposed system. During
setup and registration, IoT devices generate key pairs and certificate-signing
requests, after which the certification and registration authority issues
identity-bound certificates and publishes the public parameters required for
protocol execution. Each device encrypts its sensed value $v_i$ as
\[
ct_i=\operatorname{Enc}(v_i;r_i)
\]
and transmits
\[
T_i=(ID_i,ct_i,w_i,ts,\mathsf{nonce},\sigma_i),
\]
where the identity, freshness, weighting, and authentication metadata support
secure transaction validation.

After identity and transaction verification, validators compute
\[
ct_{\mathrm{agg}}
=
\prod_{i=1}^{n_d}
ct_i^{w_i}
\bmod n^2,
\]
thereby performing weighted aggregation over encrypted data without disclosing
individual sensor readings. When the target function is outside the linear
computation domain, the workflow activates the optional BFV bridge within the
TEE to support subsequent homomorphic computation.

The resulting aggregate is committed to the blockchain together with block
metadata, Merkle-root integrity information, timestamps, and leader
authentication data. Network agreement follows the configured PBFT-style
consensus procedure. Finally, an authorized querier obtains the aggregate
through threshold decryption and may receive associated zero-knowledge
verification evidence, where such verification is enabled.

\begin{figure*}[t]
\centering
\includegraphics[width=\textwidth]{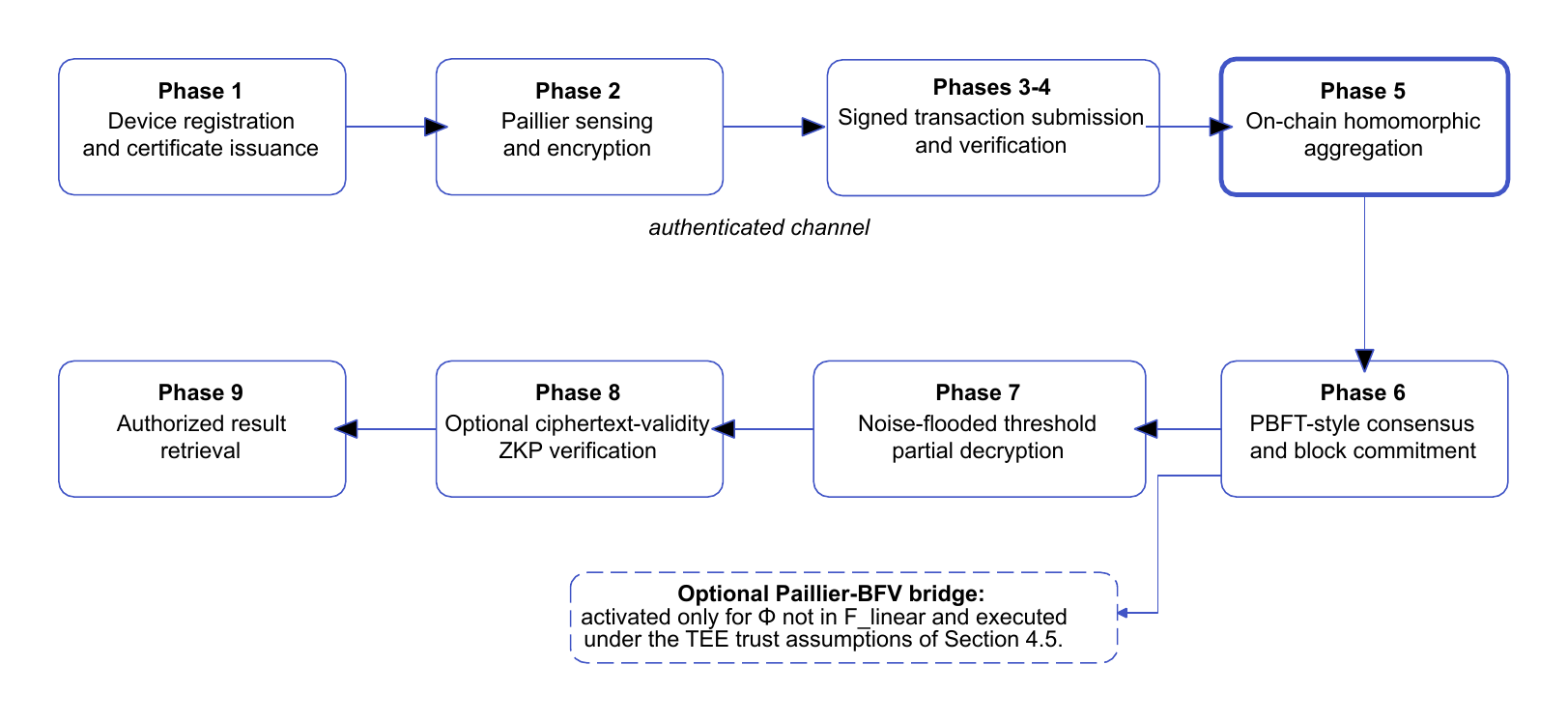}
\caption{Operational workflow of the evaluated $\Phi$-PHE-BC architecture. The present implementation uses classical transaction authentication and key establishment}
\label{fig:workflow}
\end{figure*}
\section{Security Analysis}
\label{sec:security_analysis}

We establish the confidentiality and transaction-integrity properties of
every $\Phi$-PHE-BC instance satisfying
$L=\mathrm{LAYERSELECT}(D,\Phi,\lambda)$ against the adversary
$\mathcal{A}(\lambda,f,q_e,q_h)$ defined in Section~\ref{subsec:adversary}.
The security guarantees in this section are deliberately separated according
to their underlying assumptions. In particular, confidentiality of Paillier
ciphertexts and transaction integrity rely on classical computational
assumptions, whereas the statistical protection provided by the noise-flooding
mechanism is independent of the computational power of the adversary.

\setcounter{Theorem}{3}
\begin{Theorem}[Paillier IND-CPA Security]
\label{thm:paillier_ind_cpa}
For any probabilistic polynomial-time classical adversary $\mathcal{A}_0$,
the IND-CPA advantage against the Paillier cryptosystem is bounded by
\begin{equation}
\operatorname{Adv}^{\mathrm{IND\text{-}CPA}}_{\mathrm{Paillier}}(\mathcal{A}_0)
\leq
2\operatorname{Adv}^{\mathrm{DCR}}_{n}(\mathcal{A}'_0),
\label{eq:paillier_ind_cpa}
\end{equation}
where $\mathcal{A}'_0$ is an adversary against the Decisional Composite
Residuosity (DCR) assumption with comparable running time, up to the
overhead of the reduction.
\end{Theorem}

\noindent
Theorem~\ref{thm:paillier_ind_cpa} establishes the confidentiality baseline
for on-chain homomorphic aggregation. Under the DCR assumption, an efficient
classical adversary cannot distinguish encryptions of two chosen plaintexts
with non-negligible advantage. Consequently, the individual ciphertexts
$\mathit{ct}_i$ and the aggregate ciphertext $\mathit{ct}_{\mathrm{agg}}$
remain confidential against classical polynomial-time adversaries.

\begin{Theorem}[Noise-Flooding Protection]
\label{thm:noise_flooding}
Let $e_0$ and $e_1$ be any two admissible secret-dependent threshold exponents,
and let $F_{\mathrm{nf}}$ be a flooding distribution satisfying
Eq.~\eqref{eq:nf_requirement}. Define
\begin{equation}
\widehat e_b=e_b+f,\qquad f\leftarrow F_{\mathrm{nf}},\qquad b\in\{0,1\}.
\end{equation}
Then
\begin{equation}
\operatorname{SD}\!\left(D_{\widehat e_0},D_{\widehat e_1}\right)
\leq \varepsilon_{\mathrm{nf}}.
\label{eq:noise_flooding_sd}
\end{equation}
For any fixed aggregate ciphertext $ct_{\mathrm{agg}}$, let
\begin{equation}
\widehat d_b = ct_{\mathrm{agg}}^{\widehat e_b}\bmod n^2.
\end{equation}
Then the corresponding published partial-decryption distributions satisfy
\begin{equation}
\operatorname{SD}\!\left(D_{\widehat d_0},D_{\widehat d_1}\right)
\leq \varepsilon_{\mathrm{nf}}.
\end{equation}
\end{Theorem}

\begin{proof}
The first inequality is exactly the configured statistical-hiding condition of
Eq.~\eqref{eq:nf_requirement}. The map
\[
x\mapsto ct_{\mathrm{agg}}^x\bmod n^2
\]
is deterministic. By the data-processing inequality for statistical distance,
deterministic post-processing cannot increase statistical distance. Therefore,
applying this map to the two flooded exponent distributions preserves the upper
bound $\varepsilon_{\mathrm{nf}}$.
\end{proof}

\begin{Remark}
Theorem~\ref{thm:noise_flooding} is an information-theoretic privacy statement:
its conclusion is expressed in statistical distance and therefore does not
rely on DCR, ECDSA, or any computational limitation of the adversary. It applies
only to the flooded threshold-share transcript under the explicitly configured
condition in Eq.~\eqref{eq:nf_requirement}. 
\end{Remark}

\begin{Theorem}[End-to-End Classical Confidentiality Bound]
\label{thm:composable_ind_cpa}
Let
\begin{equation}
F=(G,D,\Phi,\Lambda,\Pi)
\end{equation}
be a $\Phi$-PHE-BC instance with
\begin{equation}
L=\mathrm{LAYERSELECT}(D,\Phi,\lambda).
\end{equation}
For any probabilistic polynomial-time classical adversary
$\mathcal{A}(\lambda,f,q_e,q_h)$, the confidentiality advantage is bounded by
\begin{equation}
\mathrm{Adv}^{\mathrm{FIND\text{-}CPA}}(\mathcal{A}) \leq \mathrm{Adv}^{\mathrm{nDCR}}(\mathcal{A}') + n_v \, \varepsilon_{nf} + \mathrm{negl}(\lambda).
\label{eq:composable_ind_cpa}
\end{equation}
\end{Theorem}

\begin{proof}
We consider a sequence of hybrid games.

\textbf{Hybrid $H_0$.}
This is the real $\Phi$-PHE-BC execution. The adversary submits
challenge readings $(m_0,m_1)$, and the challenger encrypts $m_b$
for a uniformly chosen bit $b\in\{0,1\}$.

\textbf{Hybrid $H_1$.}
We replace the threshold-decryption outputs associated with the
challenge execution by their noise-flooded counterparts.
By Theorem~\ref{thm:noise_flooding}, the statistical distance
between the corresponding transcript distributions is bounded by
$\varepsilon_{\mathrm{nf}}$ for each admissible threshold share.
Applying the union bound over at most $n_v$ published validator
shares gives
\begin{equation}
\left|
\Pr[H_0]-\Pr[H_1]
\right|
\leq
n_v\varepsilon_{\mathrm{nf}}.
\label{eq:hybrid_nf}
\end{equation}

If the deployment chooses $\varepsilon_{\mathrm{nf}}$ negligible
in the security parameter $\lambda$ and $n_v$ is polynomially
bounded in $\lambda$, then
$n_v\varepsilon_{\mathrm{nf}}$ is also negligible.

\textbf{Hybrid $H_2$.}
We replace the Paillier challenge ciphertext with an encryption of
zero. By Theorem~\ref{thm:paillier_ind_cpa}, any probabilistic
polynomial-time distinguisher between $H_1$ and $H_2$ can be used
to construct an adversary against the Decisional Composite
Residuosity assumption. Hence,
\begin{equation}
\left|
\Pr[H_1]-\Pr[H_2]
\right|
\leq
\operatorname{Adv}^{\mathrm{DCR}}_{n}(\mathcal{A}').
\label{eq:hybrid_dcr}
\end{equation}

In $H_2$, the challenge ciphertext is independent of the challenge
bit $b$ except with negligible probability arising from the
remaining protocol abstractions. Therefore,
\begin{equation}
\operatorname{Adv}^{\mathrm{IND\text{-}CPA}}_{\mathcal{F}}
(\mathcal{A})
\leq
\operatorname{Adv}^{\mathrm{DCR}}_{n}(\mathcal{A}')
+
n_v\varepsilon_{\mathrm{nf}}
+
\operatorname{negl}(\lambda).
\end{equation}

Thus, under the DCR assumption and the configured statistical-hiding
condition for the flooding layer, the aggregation path satisfies the
stated classical confidentiality bound.
\end{proof}

\noindent

\begin{Theorem}[Transaction Integrity]
\label{thm:transaction_integrity}
Assume that the digital signature scheme configured for the deployment is
EUF-CMA secure. Then no probabilistic polynomial-time classical adversary
can forge a valid transaction
\begin{equation}
T_i=
(ID_i,\mathit{ct}_i,w_i,ts,\mathit{nonce},\sigma_i)
\end{equation}
for an uncompromised device $d_i$, except with negligible probability in
$\lambda$.
\end{Theorem}

\begin{proof}
The signature
\begin{equation}
\sigma_i
=
\operatorname{Sign}_{sk_i}
\left(
H(ID_i\parallel \mathit{ct}_i\parallel \mathit{nonce})
\right)
\end{equation}
binds the device identity, ciphertext, and freshness value to the
transaction. A successful forgery for an uncompromised signing key therefore
implies a successful forgery against the EUF-CMA security of the configured
digital signature scheme.

In addition, replay is prevented by checking that
\begin{equation}
\mathit{nonce}\notin\mathrm{Ledger},
\end{equation}
timestamp validation enforces the freshness condition
\begin{equation}
|ts-t_{\mathrm{now}}|
\leq
\tau_{\mathrm{ttl}},
\end{equation}
and certificate verification binds the public key to the registered device
identity. Therefore, an adversary cannot create or replay an accepted
transaction without either forging a valid signature or violating one of the
corresponding verification conditions.
\end{proof}

\noindent

\section{Network Resilience and Performance Analysis}
\label{sec:network_resilience}

We derive topology-dependent conditions and bounds for protocol liveness,
block latency, throughput, and communication cost. All results are explicit
functions of graph-theoretic properties of the validator subgraph $G_v$
introduced in Section~\ref{subsec:graph}. In particular, the
analysis establishes a graph-connectivity condition for Byzantine liveness
and quantifies the effect of validator-network diameter on latency and
throughput.

The analysis distinguishes the topology of the IoT device layer from that of
the validator layer. A device deployment may follow a tree, star, mesh, or
scale-free topology, while the validator subgraph can be independently
augmented to satisfy the connectivity required for Byzantine fault tolerance.
Thus, low-connectivity device-side topologies do not by themselves determine
the liveness of the validator consensus protocol.

\begin{Definition}[Protocol Liveness]
\label{def:protocol_liveness}
A $\Phi$-PHE-BC instance achieves \emph{liveness} if, for every admissible
execution containing at most $f$ Byzantine validators, every transaction
submitted by an honest device is eventually included in a committed block
after a finite number of BFT rounds, under the assumed partial-synchrony
model.
\end{Definition}

\setcounter{Theorem}{8}
\begin{Theorem}[Topology--Liveness Equivalence]
\label{thm:topology_liveness}
A $\Phi$-PHE-BC instance achieves liveness under adversary
$\mathcal{A}(\lambda,f,q_e,q_h)$ if and only if
\begin{equation}
\kappa(G_v)\geq f+1,
\label{eq:liveness_connectivity}
\end{equation}
where $\kappa(G_v)$ denotes the vertex connectivity of the validator
subgraph. The result holds for the considered deployment topology classes
under the authenticated-path routing model, provided that the validator
subgraph itself satisfies Eq.~\eqref{eq:liveness_connectivity}.
\end{Theorem}

\begin{proof}
We prove necessity and sufficiency.

\paragraph{Necessity.}
Suppose that
\begin{equation}
\kappa(G_v)\leq f.
\end{equation}
By the definition of vertex connectivity, there exists a vertex cut
$S\subseteq V_v$ such that
\begin{equation}
|S|\leq f
\end{equation}
and removing $S$ disconnects $G_v$ into at least two nonempty components,
denoted by $C_1$ and $C_2$.

Because the adversary controls at most $f$ validators, it can corrupt every
validator in $S$. The corrupted validators can suppress, delay, or equivocate
on protocol messages that would otherwise pass through the cut. Since $S$ is
a vertex cut, after its removal there is no honest communication path between
$C_1$ and $C_2$.

Consequently, validators in the disconnected components cannot reliably
exchange the messages required to establish the necessary BFT quorum across
the partition. Therefore, there exists an admissible execution with at most
$f$ Byzantine validators in which progress can be prevented indefinitely.
Hence liveness requires
\begin{equation}
\kappa(G_v)\geq f+1.
\end{equation}

\paragraph{Sufficiency.}
Suppose that
\begin{equation}
\kappa(G_v)\geq f+1.
\end{equation}
Then removal of any set of at most $f$ Byzantine validators leaves the
remaining honest validator subgraph connected. Since the protocol assumes
\begin{equation}
n_v\geq 3f+1,
\end{equation}
at least
\begin{equation}
n_v-f\geq 2f+1
\end{equation}
honest validators remain.

Under authenticated-path routing, every pair of honest validators can
communicate through paths consisting entirely of honest validators. Under
the assumed partial-synchrony conditions, the standard BFT quorum argument
therefore permits an honest quorum of at least $2f+1$ validators to
participate in progress and commit a block. Hence every transaction submitted
by an honest device is eventually included in a committed block.
\end{proof}

\begin{Remark}
\label{rem:validator_device_separation}
The condition $\kappa(G_v)\geq f+1$ is determined exclusively by the
validator subgraph. A tree or star deployment of IoT devices may therefore
coexist with a separately connected validator overlay satisfying the liveness
condition. Conversely, a validator subgraph that itself has
$\kappa(G_v)=1$ cannot tolerate the removal of even one critical validator
when $f\geq 1$.
\end{Remark}

\subsection{Throughput Bounds}
\label{subsec:throughput_bounds}

\begin{Definition}[Protocol Latency]
\label{def:protocol_latency}
The end-to-end latency for processing and committing one block is modeled as
\begin{equation}
L(G)
=
\operatorname{diam}(G_v)\,\overline{\tau}
+
\tau_{\mathrm{agg}}
+
\tau_{\mathrm{thresh}},
\label{eq:protocol_latency}
\end{equation}
where $\overline{\tau}$ is the mean per-hop communication latency,
$\tau_{\mathrm{agg}}$ is the homomorphic aggregation time, and
$\tau_{\mathrm{thresh}}$ is the threshold-decryption latency.
\end{Definition}

For the Paillier-based aggregation layer, the aggregation time depends on the
number of ciphertext operations and the Paillier modulus size. Accordingly,
we write
\begin{equation}
\tau_{\mathrm{agg}}
=
O(n_d\log n),
\end{equation}
where $\log n$ denotes the bit-length of the Paillier modulus. The threshold
term $\tau_{\mathrm{thresh}}$ captures the cost of generating and combining
the required $\tau_{\mathrm{thresh}}$ threshold-Paillier shares and is kept
explicit because its implementation cost depends on the selected threshold
parameters and arithmetic implementation.

\begin{Theorem}[Topology-Parameterized Throughput]
\label{thm:topology_throughput}
For a $\Phi$-PHE-BC instance with block size $B$, the throughput is
\begin{equation}
\operatorname{TPS}(G)
=
\frac{B}{L(G)}.
\label{eq:tps_general}
\end{equation}
Substituting Eq.~\eqref{eq:protocol_latency} yields
\begin{equation}
\operatorname{TPS}(G)
=
\frac{B}
{\operatorname{diam}(G_v)\overline{\tau}
+\tau_{\mathrm{agg}}
+\tau_{\mathrm{thresh}}}.
\label{eq:tps_expanded}
\end{equation}

For the considered topology classes, the corresponding asymptotic bounds are
\begin{align}
\operatorname{TPS}(T_{\mathrm{tree}})
&=
\frac{B}
{O(\log n_v)\overline{\tau}
+\tau_{\mathrm{agg}}
+\tau_{\mathrm{thresh}}},
\label{eq:tps_tree}
\\
\operatorname{TPS}(T_{\mathrm{star}})
&=
\frac{B}
{2\overline{\tau}
+\tau_{\mathrm{agg}}
+\tau_{\mathrm{thresh}}},
\label{eq:tps_star}
\\
\operatorname{TPS}(T_{\mathrm{mesh}})
&=
\frac{B}
{O(\sqrt{n_v})\overline{\tau}
+\tau_{\mathrm{agg}}
+\tau_{\mathrm{thresh}}},
\label{eq:tps_mesh}
\\
\operatorname{TPS}(T_{\mathrm{sf}})
&=
\frac{B}
{O(\log n_v)\overline{\tau}
+\tau_{\mathrm{agg}}
+\tau_{\mathrm{thresh}}}.
\label{eq:tps_sf}
\end{align}
\end{Theorem}

\begin{proof}
Equation~\eqref{eq:tps_general} follows from the definition of throughput as
the number of transactions committed per block divided by the end-to-end
block-processing time. Substituting the topology-dependent diameter values
into Eq.~\eqref{eq:protocol_latency} gives the stated bounds.

For a balanced tree,
\begin{equation}
\operatorname{diam}(T_{\mathrm{tree}})
=
O(\log n_v).
\end{equation}
For a star,
\begin{equation}
\operatorname{diam}(T_{\mathrm{star}})
=
2.
\end{equation}
For a two-dimensional mesh,
\begin{equation}
\operatorname{diam}(T_{\mathrm{mesh}})
=
O(\sqrt{n_v}).
\end{equation}
For the considered scale-free deployment model, the characteristic path
length and effective diameter scale logarithmically under the adopted
asymptotic model, giving
\begin{equation}
\operatorname{diam}(T_{\mathrm{sf}})
=
O(\log n_v).
\end{equation}
Substitution into Eq.~\eqref{eq:tps_expanded} proves the result.
\end{proof}

\noindent
A star-shaped \emph{validator} topology has the smallest diameter among the
listed structures and therefore the smallest communication component in the
latency model. However, for Byzantine fault tolerance it must still satisfy
Theorem~\ref{thm:topology_liveness}; consequently, a device-layer star should
not be interpreted as requiring a single hub validator.

\subsection{Per-Block Communication Cost}
\label{subsec:communication_cost}

\setcounter{Proposition}{10}
\begin{Proposition}[Per-Block Communication Cost]
\label{prop:communication_cost}
The total communication volume per committed block can be expressed as
\begin{equation}
C(G,n_d,n_v,B)
=
n_d|T|
+
n_v^2|\mathrm{Vote}|
+
\tau_{\mathrm{thresh}}|\widehat{d}|
+
|B_k|,
\label{eq:communication_cost}
\end{equation}
where $|T|$, $|\mathrm{Vote}|$, $|\widehat{d}|$, and $|B_k|$ denote,
respectively, the byte sizes of a transaction, a BFT vote, a noise-flooded
partial-decryption message, and the committed block metadata.
\end{Proposition}

\begin{proof}
The first term accounts for the submission of $n_d$ device transactions.
The second term captures the dominant all-to-all BFT voting component, which
requires $O(n_v^2)$ vote transmissions under the adopted broadcast model.
The third term accounts for the
$\tau_{\mathrm{thresh}}$ threshold-decryption shares required for
reconstruction, and the final term represents block metadata and associated
commitment information. Summing these contributions yields
Eq.~\eqref{eq:communication_cost}.
\end{proof}

\begin{Remark}
The dominant communication term is generally
\begin{equation}
n_v^2|\mathrm{Vote}|,
\end{equation}
which reflects the quadratic communication cost of the PBFT-style voting
procedure. A hierarchical or Merkle-based aggregation mechanism can reduce
the communication burden associated with collecting device submissions, but
it does not remove the quadratic validator-voting component of the adopted
BFT protocol.
\end{Remark}

\noindent
The results of this section establish topology-dependent guarantees for
validator connectivity, liveness, communication delay, throughput, and
per-block communication cost. The analysis does not claim, in its present
form, a separate closed-form resilience metric based on the Fiedler value
$\lambda_2(G_v)$ or the maximum degree $\Delta(G_v)$. Such a metric should be
introduced only if it is formally defined and proved in a subsequent
theorem. The experimental consequences of the topology-dependent bounds are
evaluated in Section~\ref{sec:eval}.
\section{Game-Theoretic Stability}
\label{sec:game}

We model validator participation as a strategic-form game and derive the
conditions under which honest participation is the dominant strategy. The
analysis is deliberately parameterized by a generic cryptographic-cost
multiplier so that the incentive result applies to the classical cryptographic
stack evaluated in Section~\ref{sec:eval} as well as to heavier cryptographic
stacks that may be considered in future work.
\begin{Definition}[Validator Utility]
\label{def:utility}
The utility of validator $V_i$ is
\begin{equation}
U_i(s_i)
=
\begin{cases}
R_i-\alpha C_i,
& \text{if } s_i=\mathsf{honest},\\
G_i-\alpha C_i-P_i,
& \text{if } s_i=\mathsf{deviate},
\end{cases}
\label{eq:utility}
\end{equation}
where $R_i>0$ is the block reward for honest participation, $C_i>0$ is the
computational cost of executing one block's cryptographic operations,
$\alpha\geq 1$ is a generic cryptographic-cost multiplier relative to a
lightweight baseline, $G_i\geq 0$ is the expected gain from a successful
deviation, and $P_i>0$ is the penalty imposed upon detection.
\end{Definition}

\setcounter{Theorem}{11}
\begin{Theorem}[Nash Equilibrium Stability]
\label{thm:nash}
Honest participation is the dominant strategy for every validator if and only
if
\begin{equation}
P_i>G_i-R_i,
\qquad \forall i\in\mathcal{V}_v.
\label{eq:nash_cond}
\end{equation}
Moreover, the cryptographic-cost multiplier $\alpha$ does not appear in
Eq.~\eqref{eq:nash_cond}; therefore, within the utility model of
Eq.~\eqref{eq:utility}, changing the common per-block cryptographic overhead
does not by itself require redesign of the reward-and-penalty condition.
\end{Theorem}

\begin{proof}
Honest participation is dominant if and only if
$U_i(\mathsf{honest})>U_i(\mathsf{deviate})$. Using
Eq.~\eqref{eq:utility},
\begin{align}
R_i-\alpha C_i
&>G_i-\alpha C_i-P_i,\\
R_i
&>G_i-P_i,\\
P_i
&>G_i-R_i.
\end{align}
The terms $\alpha C_i$ cancel identically, proving
Eq.~\eqref{eq:nash_cond} and the claimed independence from the common
cryptographic-cost multiplier.
\end{proof}

The minimum penalty that maintains strict preference for honest participation
is therefore
\begin{equation}
P_i^{*}=G_i-R_i+\varepsilon,
\qquad \varepsilon>0.
\label{eq:opt_penalty}
\end{equation}
Consequently, a protocol operator changing the cryptographic stack needs only
to verify that the deployed penalty still satisfies
Eq.~\eqref{eq:nash_cond}, provided that the additional cryptographic cost is
incurred symmetrically by the honest and deviating strategies as assumed in
Eq.~\eqref{eq:utility}. This qualification is important: the theorem does not
cover deviations that selectively avoid part of the prescribed cryptographic
workload.

The validator utility may also depend on network position. Under a
reward-allocation model in which a validator's effective processing
opportunity increases with its degree, we write
\begin{equation}
R_i\propto \deg_i\,B\,r_{\mathrm{unit}},
\label{eq:degree_reward}
\end{equation}
where $\deg_i$ is the degree of validator $V_i$ in $G_v$, $B$ is the block
size, and $r_{\mathrm{unit}}$ is the per-transaction reward. Under this model,
substitution of Eq.~\eqref{eq:degree_reward} into
Eq.~\eqref{eq:nash_cond} implies that a larger effective reward reduces the
minimum penalty required by Eq.~\eqref{eq:opt_penalty}. Dense validator
connectivity is independently beneficial to liveness because
Theorem~\ref{thm:topology_liveness} requires
$\kappa(G_v)\geq f+1$, as established in
Section~\ref{sec:network_resilience}. The game-theoretic argument should
therefore be interpreted as complementary to, rather than a replacement for,
the graph-theoretic resilience analysis of Section~\ref{sec:network_resilience}.

\section{Implementation and Evaluation}
\label{sec:eval}

The evaluation considers multiple benchmark configurations to assess
performance and security--performance trade-offs of the classical
$\Phi$-PHE-BC implementation. For topology and throughput analysis, the
validator population is varied over $n_v\in\{4,8,16,32,64\}$ with a fixed
device population of 200 and a block size of 100 transactions, and results are
averaged over 20 trials. The emulated network conditions include an
intra-data-center mean link latency of $5\,\mathrm{ms}$ and a wide-area mean
link latency of $50\,\mathrm{ms}$. Threshold-decryption overhead is evaluated
at $n_v=16$ and 200 devices using 1000 trials per configuration, a 2048-bit
Paillier modulus, and a nominal 128-bit classical security level. Byzantine
resilience is evaluated at $n_v=7$, for which $f=2$, with
$n_{\mathrm{byz}}\in\{0,1,2\}$. Thus, the empirical Byzantine experiments
remain inside the configured fault-tolerant regime. Behavior outside that
regime is treated analytically through
Theorem~\ref{thm:topology_liveness}, rather than being inferred from
unperformed partition-attack experiments.

\begin{table*}[t]
\centering
\caption{Unified benchmark configurations used in the evaluation. Distinct validator sets correspond to distinct experimental objectives and should not be interpreted as one common sweep.}
\label{tab:benchmark_config}
\begin{tabular}{p{0.24\textwidth}p{0.20\textwidth}p{0.17\textwidth}p{0.29\textwidth}}
\toprule
\textbf{Experiment} & \textbf{Network size} & \textbf{Trials / latency} & \textbf{Fixed parameters and objective}\\
\midrule
Topology scaling & $n_v\in\{4,8,16,32,64\}$; $n_d=200$ & 20 trials; $\overline{\tau}=5$ or $50$ ms & Block size $B=100$; topology-dependent latency/throughput sensitivity \\
PBFT validator sensitivity & $n_v\in\{4,5,6,7,8,10,12,16\}$ & 20 trials; $\overline{\tau}=5$ ms & Measured TPS series reported in Table~\ref{tab:throughput_validators} \\
Threshold decryption & $n_v=16$; $n_d=200$ & 1000 trials & 2048-bit Paillier modulus; nominal 128-bit classical security level \\
Byzantine-load sensitivity & $n_v=7$, $f=2$; $n_{\mathrm{byz}}\in\{0,1,2\}$ & Recorded fault-tolerant trials & Consensus latency and decryption success within the admissible fault regime \\
\bottomrule
\end{tabular}
\end{table*}

\subsection{Implementation}
The experiments in this section evaluate the primary Paillier threshold-aggregation path with classical ECDSA-256 transaction authentication. The optional BFV/CKKS bridge and optional ZKP configurations described in Section~\ref{sec:framework} are architectural extensions and are not benchmarked as part of the reported baseline unless explicitly stated.

\label{subsec:impl}

The experimental environment is built on Ubuntu 22.04 LTS (x86-64) with
Python 3.11 as the primary orchestration and experiment language.
Containerization and orchestration use Docker (v25.x) and Docker Compose
(v2.x). The blockchain layer uses Hyperledger Fabric v2.5.4, Fabric CA v1.5.9
for identity management, and CouchDB v3.3.x as the world-state database.
Classical cryptographic operations, including Paillier, ECDSA/ECIES, and HKDF,
are implemented using \texttt{pycryptodome} (v3.20.0), together with SHA3-256
and SHAKE-256 from Python's standard \texttt{hashlib} library. Graph analysis
and scientific computation use NetworkX (v3.3), NumPy (v1.26.4), and SciPy
(v1.13.0); visualization uses Matplotlib (v3.8.4), and data handling uses
Pandas (v2.2.2). Network conditions are emulated using \texttt{tc netem}
(iproute2 6.x) with normally distributed per-link delay and $1\,\mathrm{ms}$
jitter, supported by Linux network namespaces. Testing uses \texttt{pytest}
(v8.1.1), and Fabric client interactions use \texttt{grpcio} and
\texttt{grpcio-tools} (v1.63.0).

\subsubsection{Chaincode}
The PHE aggregation chaincode is implemented in Go and exposes nine
transaction functions corresponding to the nine protocol phases described in
Section~\ref{subsec:protocol}. Homomorphic aggregation is performed using
Go's \texttt{math/big} package for modular arithmetic. Threshold partial
decryption and Lagrange combination are also integrated into the chaincode
workflow. The implementation also includes the mask-cancellation step used by the
noise-flooding wrapper, so that the aggregate masking contribution is
removed during threshold reconstruction before final plaintext
recovery. Therefore, the measured threshold-decryption results include
the correctness-preserving flooding/cancellation procedure rather than
an unflooded threshold-Paillier baseline. The information-theoretic protection associated with this flooding step is the property established in
Theorem~\ref{thm:noise_flooding}; it should not be conflated with the
classical computational security of Paillier.

\subsubsection{Transaction Authentication}
Device transactions are authenticated using ECDSA-256 signatures in the
evaluated implementation, consistent with the classical EUF-CMA transaction
integrity result of Theorem~\ref{thm:transaction_integrity}. The baseline is a
classical PHE--blockchain configuration using Paillier aggregation,
ECDSA-256 signatures, and ECDH key establishment, corresponding to the
classical design family represented by Si et al.~\cite{SI202468}. 

\subsubsection{Topology Emulation}
The four topology classes are emulated using Linux network namespaces and
\texttt{tc netem} for per-link latency injection. The mean link latency
$\overline{\tau}$ is configured as $5\,\mathrm{ms}$ for the intra-data-center
baseline and $50\,\mathrm{ms}$ for the wide-area IoT setting. For topology-scaling experiments, validator counts are selected from
$n_v\in\{4,8,16,32,64\}$; the separate PBFT validator-sensitivity experiment
uses the set reported in Table~\ref{tab:benchmark_config}. Device counts are
selected from $n_d\in\{50,200,500\}$. These experiments operationalize the
network-diameter and connectivity effects analyzed in
Section~\ref{sec:network_resilience}.

The proposed $\Phi$-PHE-BC implementation records an encryption latency of
$7.2\,\mathrm{ms}$ in the comparison summarized in
Table~\ref{tab:enc_latency}. Figure~\ref{fig:crypto_performance_group}(a) visualizes the same
comparison. Figure~\ref{fig:crypto_performance_group}(b) shows the effect of increasing batch
size on per-transaction encryption cost, while
Figure~\ref{fig:crypto_performance_group}(c) illustrates the ciphertext-size expansion of
Paillier relative to ECIES. Figure~\ref{fig:crypto_performance_group}(d) reports threshold
decryption latency as a function of the threshold parameter $t$; in the tested
range, the latency remains below $20\,\mathrm{ms}$.

\begin{table}
\centering
\caption{Encryption latency comparison across representative PHE-based systems (lower is better).}
\label{tab:enc_latency}
\begin{tabular}{lc}
\toprule
\textbf{System} & \textbf{Encryption time (ms)}\\
\midrule
Si et al. & 7.8\\
Aronoff et al. & 15.0\\
Njungle et al. & 8.0\\
$\Phi$-PHE-BC (Proposed) & \textbf{7.2}\\
\bottomrule
\end{tabular}
\end{table}

\begin{figure*}[t]
\centering
\includegraphics[width=0.235\linewidth]{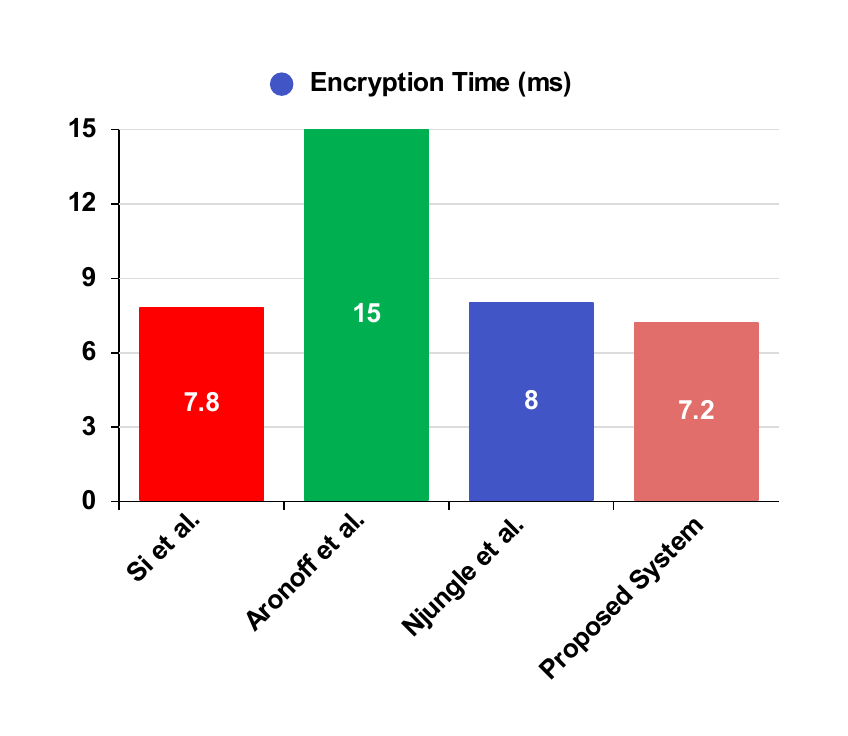}\hfill
\includegraphics[width=0.235\linewidth]{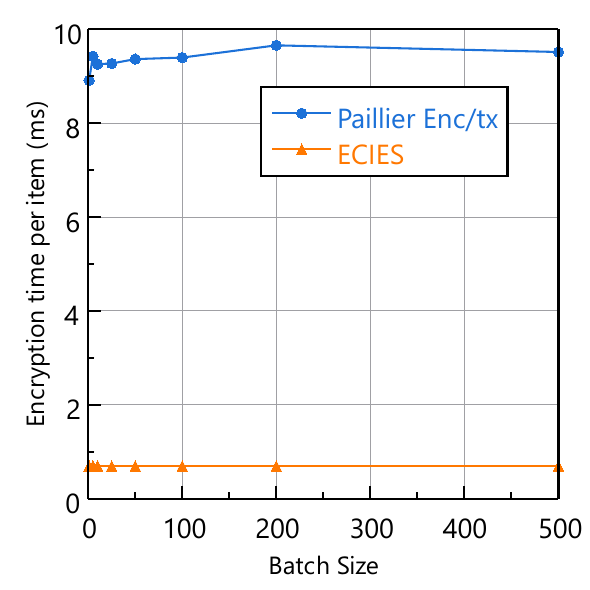}\hfill
\includegraphics[width=0.235\linewidth]{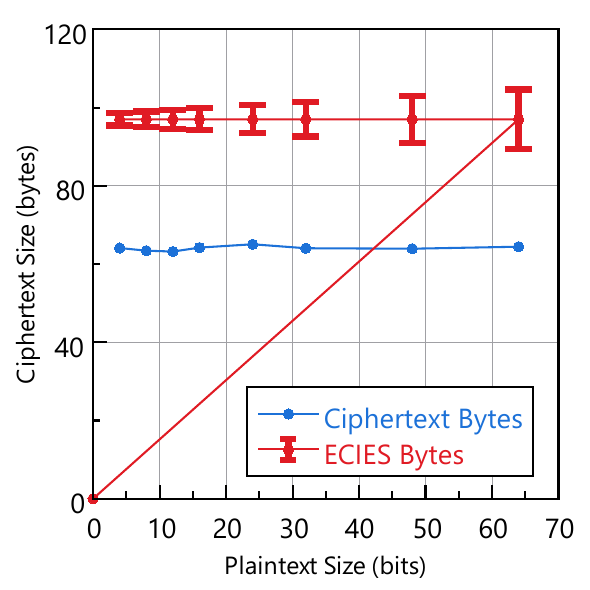}\hfill
\includegraphics[width=0.235\linewidth]{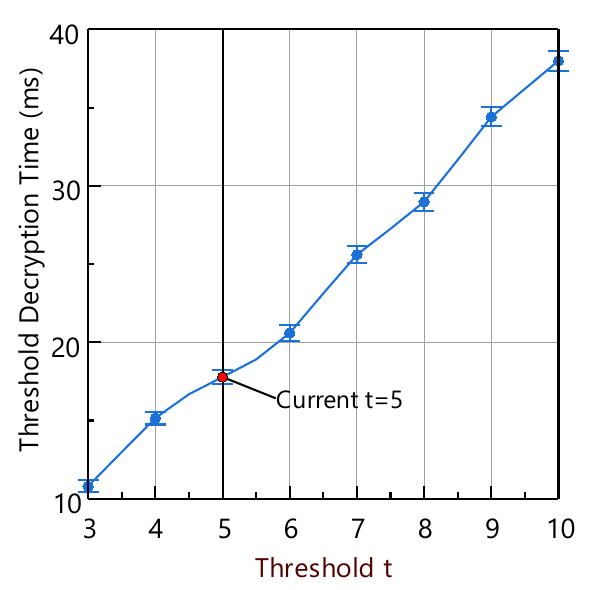}
\caption{Cryptographic performance of the evaluated $\Phi$-PHE-BC implementation: (a) encryption-time comparison across representative PHE-based systems; (b) encryption time per transaction under increasing batch size; (c) ciphertext-size comparison between Paillier and ECIES; and (d) threshold-decryption latency as a function of threshold parameter $t$.}
\label{fig:crypto_performance_group}
\end{figure*}

\subsection{Throughput and Latency}
\label{subsec:bench_tps}

Measured throughput under the intra-data-center latency setting
($\overline{\tau}=5\,\mathrm{ms}$) is consistent with the topology dependence
predicted by Theorem~\ref{thm:topology_throughput}: topologies with smaller
effective validator-network diameter incur a smaller communication component
in Eq.~\eqref{eq:protocol_latency}. Across the validator-count benchmark,
throughput ultimately decreases as the validator population grows because the
adopted BFT voting procedure contains a quadratic all-to-all communication
component, as quantified by Proposition~\ref{prop:communication_cost}.
Table~\ref{tab:throughput_validators} reports the measured PBFT-style
throughput series, and Figure~\ref{fig:consensus_performance_group}(a) shows the
corresponding scaling trend.

\begin{table*}
\centering
\caption{Throughput under increasing validator count (PBFT consensus).}
\label{tab:throughput_validators}
\begin{tabular}{lcccccccc}
\toprule
\textbf{$n_v$} & \textbf{4} & \textbf{5} & \textbf{6} & \textbf{7} & \textbf{8} & \textbf{10} & \textbf{12} & \textbf{16}\\
\midrule
Throughput (TPS) & 327 & 336 & 313 & 297 & 275 & 255 & 204 & 143\\
\bottomrule
\end{tabular}
\end{table*}

\begin{figure*}[t]
\centering
\includegraphics[width=0.48\linewidth]{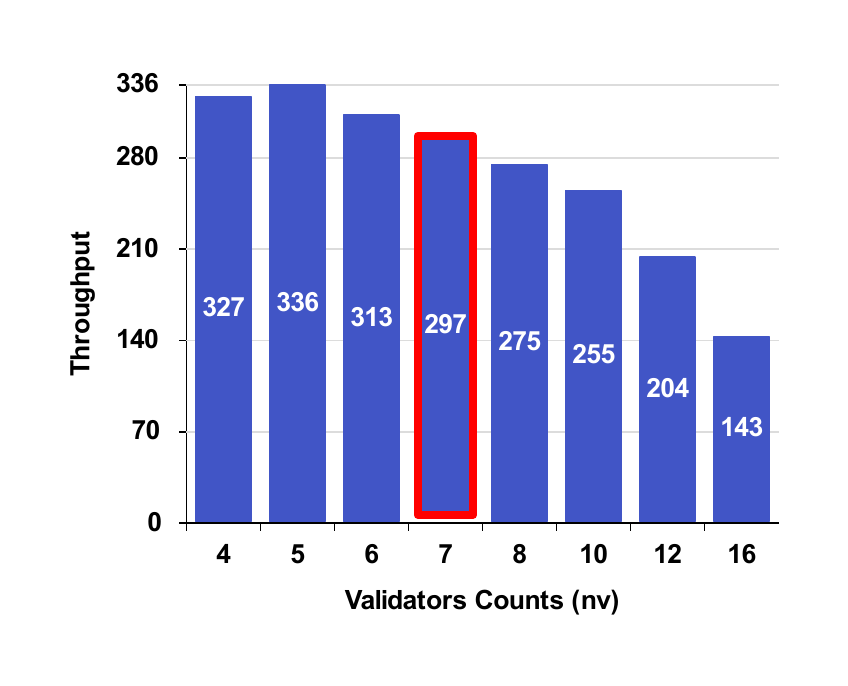}\hfill
\includegraphics[width=0.48\linewidth]{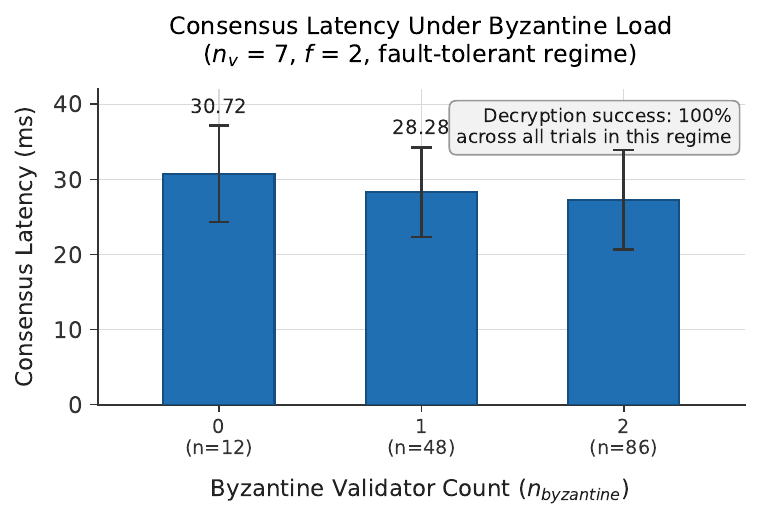}

\caption{Consensus-performance sensitivity: (a) PBFT-style throughput under increasing validator count; and (b) consensus latency under increasing Byzantine-validator count within the configured fault-tolerant regime $n_{\mathrm{byz}}\leq f$.}
\label{fig:consensus_performance_group}
\end{figure*}

For $n_v=7$ and $f=2$, consensus latency under the tested sub-threshold
Byzantine load remained between approximately $27$ and $31\,\mathrm{ms}$ for
$n_{\mathrm{byz}}\in\{0,1,2\}$, while decryption success remained at 100\%
across the recorded trials. These observations are consistent with the
sufficiency direction of Theorem~\ref{thm:topology_liveness} for the tested
fully connected validator overlay. Figure~\ref{fig:consensus_performance_group}(b) reports the
measured Byzantine-load sensitivity. No empirical claim is made for
$n_{\mathrm{byz}}>f$.

Transaction-per-second scaling with network size is shown in
Figure~\ref{fig:scalability_performance_group}(a). Under the benchmark used for the cross-system
comparison, the proposed system records approximately 3120 TPS compared with
approximately 1950 TPS for Si et al., while the corresponding end-to-end
transaction latency decreases from approximately $2307\,\mathrm{ms}$ to
approximately $410\,\mathrm{ms}$. The latency reduction is approximately
82\%. These values are configuration-specific benchmark results and should not
be interpreted as topology-independent deployment guarantees.
Figure~\ref{fig:scalability_performance_group}(b) shows the measured latency trend as the number
of IoT devices increases; Figure~\ref{fig:scalability_performance_group}(c)--(d) provide the cross-system latency and throughput
comparisons.

\begin{figure*}[t]
\centering
\includegraphics[width=0.235\linewidth]{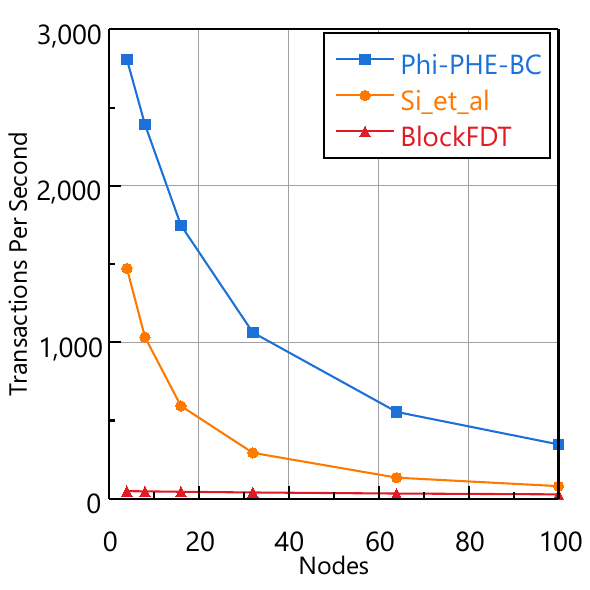}\hfill
\includegraphics[width=0.235\linewidth]{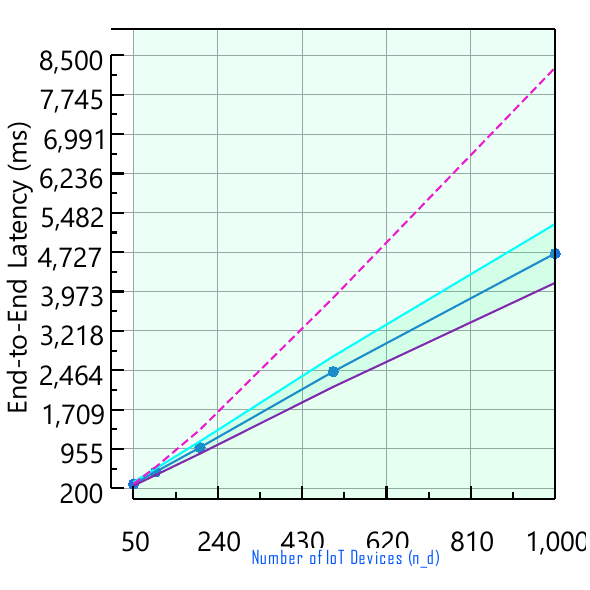}\hfill
\includegraphics[width=0.235\linewidth]{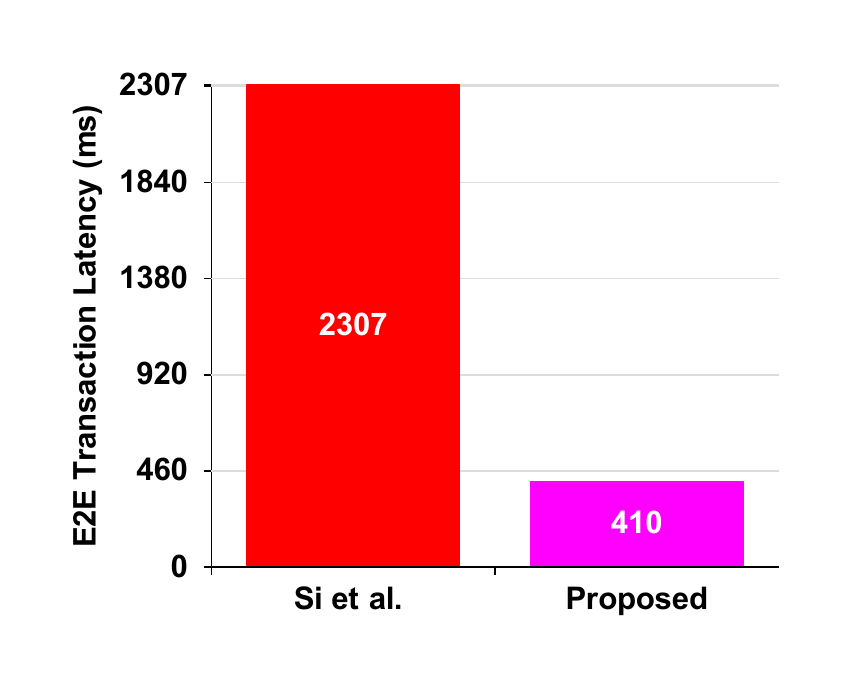}\hfill
\includegraphics[width=0.235\linewidth]{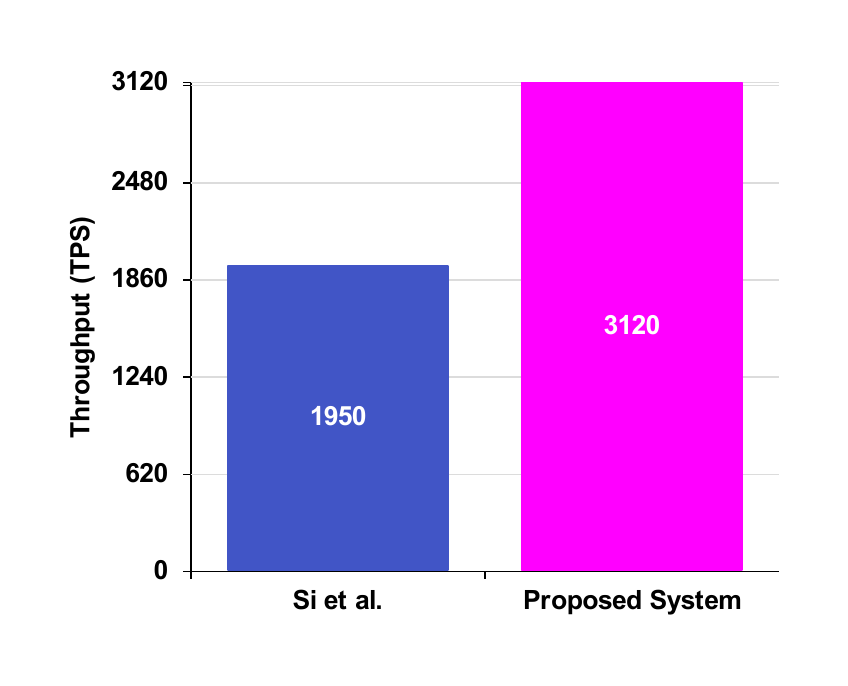}

\caption{Scalability and cross-system performance: (a) transaction-per-second scaling with network-node count; (b) end-to-end latency with increasing IoT-device population; (c) end-to-end transaction-latency comparison against the selected prior system; and (d) throughput comparison under the stated cross-system benchmark configuration.}
\label{fig:scalability_performance_group}
\end{figure*}

\begin{table}
\centering
\caption{End-to-end latency and throughput comparison against the selected prior system.}
\label{tab:e2e_comparison}
\begin{tabular}{lcc}
\toprule
\textbf{Metric} & \textbf{Si et al.} & \textbf{$\Phi$-PHE-BC}\\
\midrule
End-to-end transaction latency (ms) & 2307 & 410\\
Throughput (TPS) & 1950 & 3120\\
Threshold decryption time (ms) & 9.8 & 13.8\\
\bottomrule
\end{tabular}
\end{table}

The threshold-decryption entry in Table~\ref{tab:e2e_comparison} shows that
$\Phi$-PHE-BC incurs a higher threshold-decryption time than the selected Si et
al. comparison value ($13.8\,\mathrm{ms}$ versus $9.8\,\mathrm{ms}$), even
though its reported end-to-end latency and throughput are better under the
cross-system benchmark. This is the expected trade-off to report from the
numerical values provided, rather than claiming a threshold-decryption speedup.
Figure~\ref{fig:baseline_performance_group}(a) visualizes this threshold-decryption
comparison.

Figure~\ref{fig:baseline_performance_group}(b) compares plaintext blockchain, single-key PHE,
ECIES-based encrypted blockchain, and the proposed threshold-PHE blockchain
configuration. Plaintext operation minimizes cryptographic overhead but
provides no confidentiality; a single-key PHE design retains a centralized
key-recovery point; and the proposed threshold design distributes decryption
authority at the cost of additional cryptographic and communication overhead.

\begin{figure*}[t]
\centering
\includegraphics[width=0.40\linewidth]{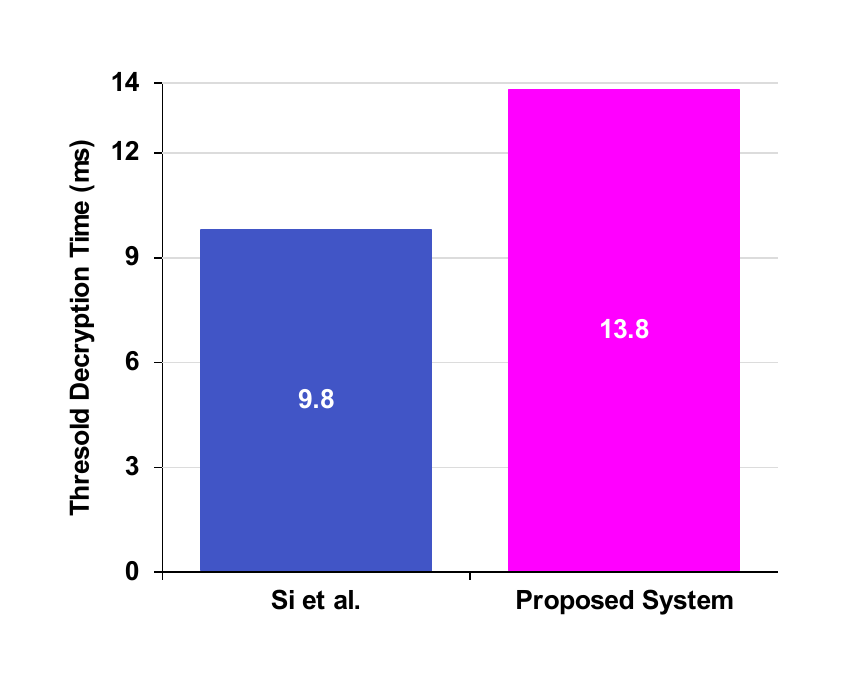}\hfill
\includegraphics[width=0.56\linewidth]{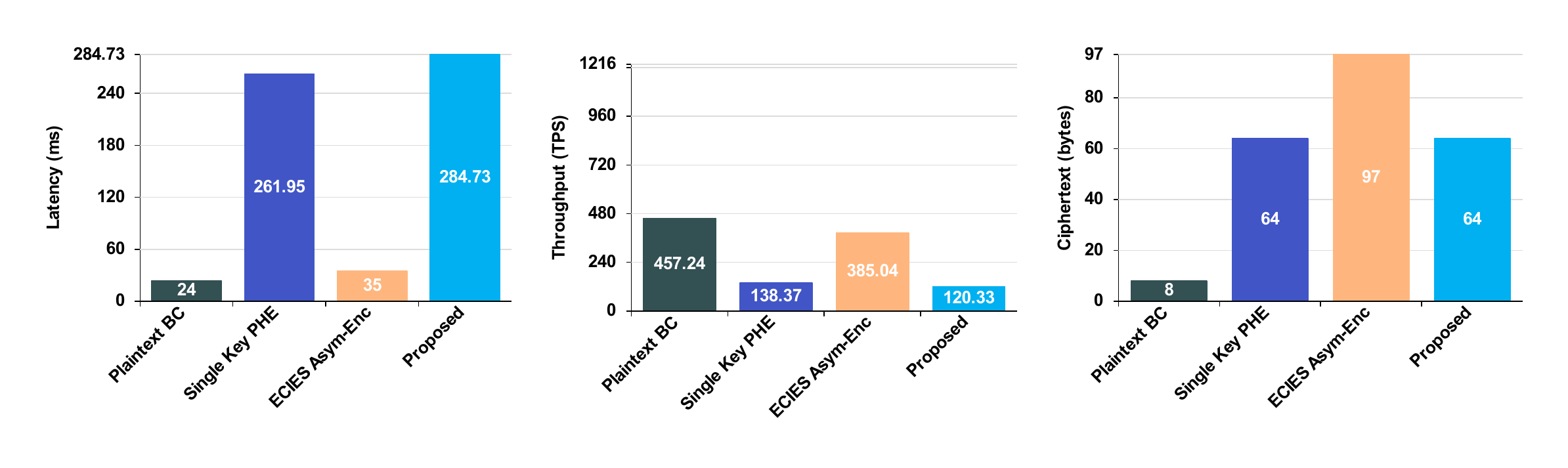}

\caption{Detailed baseline comparisons: (a) threshold-decryption time against the selected prior framework; and (b) latency--throughput--ciphertext trade-off across the evaluated blockchain variants.}
\label{fig:baseline_performance_group}
\end{figure*}

\subsection{Security--Performance Comparison}
\label{subsec:security_performance}

Table~\ref{tab:security_matrix} compares the architectural properties reported
for the selected schemes. The distinguishing feature of $\Phi$-PHE-BC in this
comparison is the joint use of threshold decryption, blockchain integration,
homomorphic aggregation, topology-aware liveness, throughput, and connectivity-based fault-tolerance analysis.

\begin{table*}[t]
\centering
\caption{Security and feature matrix comparison.}
\label{tab:security_matrix}
\renewcommand{\arraystretch}{1.15}
\begin{tabular}{lccccc}
\toprule
\textbf{Scheme} & \textbf{Threshold} & \textbf{Blockchain} & \textbf{Homomorphic} & \textbf{Topology-Aware} & \textbf{IoT}\\
& \textbf{Decryption} & \textbf{Integration} & \textbf{Encryption} & \textbf{Liveness} & \textbf{Support}\\
\midrule
Block-FDT & \xmark & \cmark & \xmark & \xmark & \cmark\\
Si et al. & \cmark & \cmark & \cmark & \xmark & \xmark\\
Ci et al. & \xmark & \xmark & \cmark & \xmark & \xmark\\
Aronoff et al. & \xmark & \cmark & \cmark & \xmark & \xmark\\
Njungle et al. & \xmark & \xmark & \cmark & \xmark & \xmark\\
$\Phi$-PHE-BC (Proposed) & \cmark & \cmark & \cmark & \cmark & \cmark\\
\bottomrule
\end{tabular}
\end{table*}

\subsection{Liveness Verification}
\label{subsec:liveness_verification}

Theorem~\ref{thm:topology_liveness} establishes that liveness holds if and
only if
\begin{equation}
\kappa(G_v)\geq f+1,
\end{equation}
under the authenticated-path and partial-synchrony assumptions stated in
Section~\ref{sec:network_resilience}. The empirical evaluation covers only the
fault-tolerant regime. For $n_v=7$ ($f=2$) and
$n_{\mathrm{byz}}\in\{0,1,2\}$, the fully connected validator overlay
satisfies the required vertex-connectivity condition, and consensus completed
successfully in 100\% of the recorded trials. The measurements summarized in
Figure~\ref{fig:consensus_performance_group}(b) are therefore consistent with the sufficiency
direction of Theorem~\ref{thm:topology_liveness}.

We did not empirically evaluate configurations violating the connectivity
condition, such as a validator overlay with $\kappa(G_v)<f+1$, nor did we run
partition-attack trials for $n_{\mathrm{byz}}>f$. Consequently, the necessity
direction of Theorem~\ref{thm:topology_liveness} is supported in this paper by
the constructive partition argument of Section~\ref{sec:network_resilience},
not by an empirical failure-regime experiment. Directly validating that
failure regime across multiple values of $n_v$ and multiple validator
connectivity patterns is retained as future work.

\subsection{Limitations and Future Work}
\label{subsec:limitations_future}

\subsubsection{Experimental Scope}
The experimental evaluation uses an emulated Hyperledger Fabric environment
with controlled latency settings and modeled topology classes rather than a
large-scale physical IoT deployment. The reported latency and throughput
values should therefore be interpreted as implementation-level evidence under
the tested configurations, not as universal deployment guarantees. The
experiments cover selected validator counts, device populations, and workload
settings; larger networks, heterogeneous edge hardware, intermittent links,
and mobility can alter both cryptographic and consensus costs.

\subsubsection{Trusted-Execution Dependency}
The optional Paillier--BFV bridge of Section~\ref{subsec:bridge} assumes a
trusted execution environment for secure decrypt-and-re-encrypt processing.
The practical security of that optional component therefore, depends on the
TEE trust model and implementation assumptions stated in the framework design.
This dependency is separate from the security of the primary Paillier
aggregation path evaluated in this section.

\subsubsection{Device-Layer Topology Augmentation}
Tree- and star-shaped device deployments do not by themselves establish the
validator-connectivity condition required by
Theorem~\ref{thm:topology_liveness}. The validator overlay must be configured
or augmented independently so that $\kappa(G_v)\geq f+1$, consistent with the
validator/device-topology separation discussed in
Section~\ref{sec:network_resilience}.

\subsubsection{Empirical Liveness-Failure Regime}
As discussed in Section~\ref{subsec:liveness_verification}, the empirical
liveness study covers only $n_{\mathrm{byz}}\leq f$. The failure regime
predicted by the necessity direction of
Theorem~\ref{thm:topology_liveness} remains to be validated experimentally
using controlled vertex-cut and partition attacks.

\subsubsection{Post-Quantum Migration as Future Work}
The present $\Phi$-PHE-BC architecture does \emph{not} claim general
post-quantum security. Theorem~\ref{thm:paillier_ind_cpa} provides classical
IND-CPA confidentiality for Paillier under the DCR assumption, and
Theorem~\ref{thm:transaction_integrity} provides classical transaction
integrity for the deployed ECDSA-256 authentication scheme. In contrast,
Theorem~\ref{thm:noise_flooding} provides a statistical,
information-theoretic protection result for the modeled flooded
partial-decryption shares and therefore does not rely on computational
hardness assumptions. A future quantum-capable adversary could nevertheless
threaten the long-term confidentiality of Paillier ciphertexts and the
unforgeability of ECDSA-authenticated transactions. Closing these gaps is a
future-work objective rather than a property of the current system.

A concrete migration roadmap is as follows:
\begin{itemize}
    \item \textbf{Transport and authentication.}
    Replace classical session-key establishment and ECDSA-256 transaction
    signatures with standardized post-quantum key establishment and digital
    signatures, such as ML-KEM and ML-DSA, respectively, in a future
    implementation~\cite{NIST-FIPS203,NIST-FIPS204}. Their measured effect on
    the full $\Phi$-PHE-BC transaction path must be established experimentally
    rather than assumed. Within the symmetric-cost utility model of
    Theorem~\ref{thm:nash}, any resulting common cost increase is represented
    by $\alpha$ and cancels from the honest-versus-deviate condition.

    \item \textbf{Aggregation layer.}
    Investigate replacing Paillier as the primary aggregation primitive with
    a lattice-based homomorphic scheme, such as BFV or another suitable
    Module-LWE-based construction, so that the aggregation layer itself no
    longer depends on the DCR assumption. The current Paillier-BFV bridge in
    Section~\ref{subsec:bridge} is optional and TEE-dependent; promoting a
    lattice-based primitive to the primary path would therefore require a new
    protocol design and new security analysis rather than a simple relabeling
    of the existing architecture.

    \item \textbf{Threshold decryption over the new ciphertext space.}
    Re-derive Phase~7 for the selected lattice-based ciphertext space and
    determine whether an analog of the current threshold-share protection
    can be constructed. In particular, the noise-flooding analysis of
    Theorem~\ref{thm:noise_flooding} cannot be transferred automatically from
    the present Paillier-based exponent/share model to a ring-based scheme.

    \item \textbf{End-to-end security against a quantum-capable adversary.}
    Revisit Theorem~\ref{thm:composable_ind_cpa} after the transport,
    authentication, and aggregation primitives have actually been replaced by
    post-quantum candidates. The present theorem is explicitly a classical
    end-to-end confidentiality result and should not be interpreted as a
    quantum-security proof.

    \item \textbf{Empirical post-quantum overhead.}
    Implement the migrated stack and repeat the benchmark methodology of
    Sections~\ref{subsec:impl} and~\ref{subsec:bench_tps} to measure end-to-end
    latency, throughput, ciphertext expansion, signature size, key-establishment
    cost, and validator CPU/memory overhead under realistic IoT workloads.
\end{itemize}

The roadmap above is intentionally separated from the evaluated results. No
ML-KEM, ML-DSA, BFV-primary aggregation, or other lattice-based post-quantum
configuration is treated as implemented or proven in the present paper.

\subsubsection{Variation Across Reported Throughput Results}
Throughput values in this section correspond to different evaluation
objectives. Table~\ref{tab:throughput_validators} and
Figure~\ref{fig:consensus_performance_group}(a) report validator-count sensitivity at a
fixed moderate block configuration, whereas Table~\ref{tab:e2e_comparison}
and Figure~\ref{fig:scalability_performance_group} report the cross-system
benchmark configuration. Figure~\ref{fig:baseline_performance_group}(b) reports a
fixed cross-variant trade-off configuration. These values should therefore be
compared only within their stated benchmark settings rather than treated as
interchangeable estimates of a single deployment-wide throughput constant.

\section{Conclusion}

This paper presented $\Phi$-PHE-BC, a topology-aware homomorphic blockchain
architecture for privacy-preserving IoT data aggregation. The framework combines
on-chain Paillier-based aggregation, threshold decryption with
information-theoretic protection of partial-decryption shares, standard
signature-based transaction integrity, and topology-parameterized analysis of liveness and throughput, together with an
explicit per-block communication-cost model.

The central contribution of this work is the proof that validator-graph
connectivity is a necessary and sufficient determinant of protocol liveness,
namely,
\begin{equation}
    \kappa(G_v) \geq f + 1,
\end{equation}
together with topology-parameterized throughput bounds for tree, star, mesh,
and scale-free deployment models and an explicit per-block communication-cost
expression---a treatment absent from the selected PHE-blockchain literature. We further showed that validator incentive stability,
formalized through a Nash-equilibrium analysis of honest-versus-deviate
strategies, is robust to the magnitude of the deployed cryptographic cost. This
property is directly relevant to future increases in cryptographic overhead.

The implementation results on Hyperledger Fabric suggest that the proposed
design is practical for the tested configurations and can achieve lower latency
than the selected classical baseline while preserving threshold-based distributed
trust. At the same time, these findings should be interpreted in light of the
study's assumptions, including emulated network conditions, selected topology
classes and validator sizes, and reliance on a trusted execution environment for
the optional Paillier--BFV bridge.

We deliberately scope the security claims of this paper to the classical and
information-theoretic guarantees established in Section~5. Section~\ref{subsec:limitations_future} instead outlines a concrete migration path
toward a fully post-quantum-confidential instantiation of the framework, which
constitutes the primary direction for future work.

Future research will also extend the evaluation to real-world IoT deployments,
broader workload and mobility settings, and stronger verification of bridge and
computation correctness under less trusted execution assumptions.\\






\bibliographystyle{unsrt}
\bibliography{sample}

\end{document}